\documentclass[11pt,reqno]{article}
\usepackage[left=1in,right=1in,top=0.8in,bottom=0.8in,includehead,includefoot]{geometry}
\usepackage{amsmath,amssymb,amsthm}
\usepackage{graphicx}
\usepackage{booktabs,tabularx}
\usepackage{placeins,flafter}
\usepackage{needspace}
\usepackage[round,authoryear]{natbib}
\usepackage{setspace}
\usepackage{soul}
\usepackage[dvipsnames]{xcolor}

\newtheorem{lemma}{Lemma}
\newtheorem{corollary}{Corollary}
\newtheorem{proposition}{Proposition}
\newtheorem{assumption}{Assumption}

\newcommand{\Vs}{\scriptscriptstyle V}
\newcommand{\E}{\mathbb E}
\newcommand{\PP}{\mathbb P}
\newcommand{\RR}{\mathbb R}
\newcommand{\diff}{\mathrm d}
\newcommand{\Var}{\operatorname{Var}}
\newcommand{\Cov}{\operatorname{Cov}}
\newcommand{\Corr}{\operatorname{Corr}}
\newcommand{\sgn}{\operatorname{sgn}}
\newcommand{\eps}{\varepsilon}
\newcommand{\sigeps}{\sigma_{\eps}}

\renewcommand{\leq}{\leqslant}
\renewcommand{\geq}{\geqslant}

\usepackage{hyperref}
\hypersetup{
  colorlinks=true,
  linkcolor=blue,
  citecolor=Brown,
  pdftitle={Decision Noise in Threshold Selection},
  pdfauthor={Ye Lu and Joeri Smits}
}

\begin{document}

\title{\vspace{-1.25cm}Decision Noise in Threshold Selection}
\author{Ye Lu\thanks{School of Economics, University of Sydney. Email: ye.lu1@sydney.edu.au}
\and
Joeri Smits\thanks{Department of Economics, Tufts University. Email: joeri.smits@tufts.edu}}
\date{October 2026}
\maketitle

\setstretch{1.1}
\setlength{\parskip}{0.8ex plus0.3ex minus0.3ex}

\begin{abstract}
  Can noisier threshold selection produce a group with higher mean true value? We
  answer this question by varying the scale of a fixed assessment error. Holding the
  selection rate fixed, less noise yields a first-order stochastically dominant
  distribution of selected values. At a fixed threshold, however, more noise can
  raise the selected mean even as selection loss, which weights selection mistakes by
  the distance of true value from the threshold, increases. Under joint Gaussianity,
  we characterize when the selected mean rises or falls with noise. In treatment
  selection, noise alone can generate subgroup differences in participants’ average
  gains despite identical population gain distributions.
\end{abstract}



\section{Introduction}

We study threshold selection rules based on noisy assessments. These rules arise in
applications such as \emph{screening} and \emph{self-selection}.
An evaluator may admit applicants
whose scores exceed a cutoff, while an individual may invest when perceived gains
exceed a cost. We distinguish two settings. Under a \emph{fixed selection rate}, the
threshold adjusts to select a prescribed fraction; under a \emph{fixed threshold},
the selection rate adjusts. We vary noise by proportionally scaling a baseline
assessment error while holding its joint distribution with true value fixed.

We show that, at a fixed selection rate, the true values of units selected with less
decision noise first-order stochastically dominate those of units selected with more
noise. This holds for any joint distribution of true value and baseline
assessment error. At a fixed threshold, the loss in threshold-adjusted aggregate
value relative to noise-free selection equals the population average of selection
mistakes weighted by the distance between true value and the threshold. This loss is
nondecreasing in noise, yet the selected mean can rise as the selected group’s size
and composition change.

When true value and baseline assessment error are jointly
Gaussian, we derive exact conditions under which a small increase in noise raises or
lowers the selected mean, allowing any admissible correlation between them. For
nonnegative correlation, as noise increases, the selected mean either decreases
throughout, first increases and then decreases, or increases throughout. Under
independence, an interior peak exists exactly when the noise-free selection rate
exceeds a critical rate of approximately 79.95 percent; the selection rate equals
that critical rate at the peak.
Negative correlation permits a local minimum followed by a local maximum. In treatment
selection, differences in decision noise can generate subgroup differences in the
average treatment effect on the treated (ATT) despite identical population gain
distributions.

The closest study is \citet{lagziel2022}, whose Lemma 2 shows that scaling up
noise independent of true value does not increase the selected mean at a fixed
selection rate. We extend their stated mean comparison to first-order stochastic
dominance of the distribution of selected true values, allowing arbitrary dependence
between true value and baseline assessment error.
Their examples of beneficial noise add an independent noise component, changing the
error distribution’s shape. Our main contribution instead characterizes when scaling
a fixed assessment error raises the selected mean at a fixed threshold, even as
selection loss increases. Related studies examine whether stricter cutoffs improve
selected values \citep{lagziel2019, heinsalu2020}, or study threshold decisions when
evaluators differ in diagnostic skill and preferences \citep{chan2022selection}.

Our treatment application connects to self-selection based on individual gains
\citep{roy1951, heckman2006} and to models of choices under uncertain returns
\citep{carneiro2003, cunha2005, dhaultfoeuille2013}. \cite{harris2021} studies the
estimation of perceived investment returns, while \cite{hult2019} examines treatment
allocation using noisy signals of individual treatment effects.
The Gaussian formula for the selected mean follows standard sample-selection analysis
\citep{heckman1979}. We use this framework to isolate how decision noise changes
participant composition and mean gains while holding the population gain distribution
fixed.

\section{Threshold selection under decision noise}\label{sec:model}

Let $V$ denote true value and $\eps$ a baseline assessment error, allowing arbitrary
dependence between them. At noise scale $\alpha\geq0$, the decision-relevant assessment is
$V+\alpha\eps$, and selection follows
\begin{equation}\label{eq:model}
  D_\alpha = 1\{V+\alpha\eps>\eta(\alpha)\}.
\end{equation}
Here, $\eta(\alpha)$ is the threshold, $\alpha\eps$ is \emph{decision noise}, and
$D_\alpha=1$ indicates selection. The notation permits the threshold to adjust with
$\alpha$; its dependence on $\alpha$ is determined by the selection regime below.
We vary $\alpha$ while holding the joint distribution of $(V,\eps)$ fixed.

Let $p(\alpha)=\PP(D_\alpha=1)$ denote the \emph{selection rate}, and assume
$\E|V|<\infty$. When $p(\alpha)>0$, the \emph{selected mean} is
\begin{equation}\label{eq:m-def}
  m(\alpha):=\E(V\mid D_\alpha=1).
\end{equation}
The \emph{aggregate value}, measured per population unit, is
\begin{equation}\label{eq:A-def}
  A(\alpha):=\E(VD_\alpha).
\end{equation}
Thus $A(\alpha)=p(\alpha)m(\alpha)$ whenever $p(\alpha)>0$.

We consider two settings. In the \emph{fixed selection rate} setting, the threshold
$\eta(\alpha)$ adjusts to maintain a prescribed selection rate $p(\alpha) =
p\in(0,1)$.%
\footnote{
  If $V+\alpha\eps$ has an atom at $\eta(\alpha)$, randomized tie-breaking may be
  required:
  \(
    D_\alpha
    =
    1\{V+\alpha\eps>\eta(\alpha)\}
    +
    1\{V+\alpha\eps=\eta(\alpha)\}1\{U\leq\gamma_\alpha\},
  \)
  where $U\sim U[0,1]$ is independent of $(V,\eps)$ and $\gamma_\alpha\in[0,1]$ is
  chosen so that $\PP(D_\alpha=1)=p$; see, for example, \cite{lagziel2022}.
}
In the \emph{fixed threshold} setting, $\eta(\alpha)=\eta$ and the selection rate
$p(\alpha)$ varies with noise. 
In screening, the two settings correspond to an evaluator selecting under a binding
capacity constraint or applying a fixed eligibility standard, respectively.
Self-selection with a fixed participation cost also fits the fixed threshold
setting.

\section{Selection at a fixed rate}\label{sec:fr}

At a fixed selection rate, changes in noise affect the composition of the selected
group. The following proposition shows that reducing noise shifts the distribution of
selected true values upward in first-order stochastic dominance.

\begin{proposition}\label{prop:fosd}
  Let $(V,\eps)$ have an arbitrary fixed joint distribution. At a fixed selection
  rate $p\in(0,1)$, for any $0\leq\alpha_1<\alpha_2$,
  \begin{equation}\label{eq:general-fosd}
    V\mid D_{\alpha_1}=1 
    \text{ first-order stochastically dominates }
    V\mid D_{\alpha_2}=1.
  \end{equation}
  If, in addition, $\E|V|<\infty$, then $m(\alpha_1) \geq m(\alpha_2)$.
\end{proposition}

The mechanism parallels the argument in the proof of Lemma~2 of \citet{lagziel2022}.
There is a crossing value such that reducing noise lowers selection probabilities for
units with true values below it and raises them for units with true values above it.
With the overall selection rate fixed, selection shifts toward higher true values.
For an evaluator filling a fixed number of places, reducing noise therefore cannot
lower the share of selected units whose true values exceed any given quality
standard.

Under $\E|V|<\infty$, aggregate value $A(\alpha)=p\,m(\alpha)$ is also nonincreasing
in noise. The selected-mean and aggregate-value inequalities are strict if
conditional selection probabilities differ on a set of true values with positive
probability; see Lemma~S.1 in the Supplementary Material.

\section{Selection at a fixed threshold}\label{sec:ft}

Now fix $\eta(\alpha)=\eta$ and allow the selection rate to vary. More noise can then
raise the selected mean.

\subsection{Threshold-adjusted value and loss from decision noise}\label{sec:ft-loss}

For $\E|V-\eta|<\infty$, define the \emph{threshold-adjusted aggregate value}
\begin{equation} \label{eq:W-def}
  W(\alpha) := \E \big[(V-\eta)D_\alpha \big].
\end{equation}
Thus $W(\alpha)=A(\alpha)-\eta p(\alpha)$. If $\eta$ is a selection or participation
cost excluded from $V$, $W(\alpha)$ is population-average net surplus; for an
eligibility or quality standard, it is a performance measure. Since the
noise-free rule $D_0=1\{V>\eta\}$ maximizes $\E[(V-\eta)D]$ over binary selection
rules $D$, we measure \emph{selection loss} by $W(0)-W(\alpha)$. When
$\eta$ is a common treatment cost, $W(0)-W(\alpha)$ is the welfare regret of the
noisy rule relative to the noise-free rule $D_0$, analogous to the regret criterion
in the treatment-choice literature \citep{manski2004statistical,kitagawa2018should}.

\begin{proposition} \label{prop:ftv}
  Let $(V,\eps)$ have any fixed joint distribution and suppose
  $\E|V-\eta|<\infty$. Then, for every $\alpha\geq0$,
  \begin{equation}\label{eq:threshold-loss}
    W(0)-W(\alpha)
    =
    \E\big[|V-\eta|1\{D_\alpha\neq D_0\}\big] \geq0.
  \end{equation}
  Moreover, for $0\leq\alpha_1<\alpha_2$,
  \begin{equation}\label{eq:threshold-loss-2}
    W(\alpha_1) - W(\alpha_2)
    =
    \E\big[
      |V-\eta|\, 1\{D_{\alpha_1}=D_0,\, D_{\alpha_2}\neq D_0\}
    \big]
    \geq0.
  \end{equation}
\end{proposition}

Each mistaken inclusion or exclusion contributes $|V-\eta|$ to selection loss,
so mistakes near the threshold incur little loss. When $\eta$ is a cost, these
mistakes forgo positive net gains or incur negative ones. Scaling up the same
realized error can create new selection mistakes but cannot correct existing ones.
Thus selection loss is nondecreasing in noise, and $W(\alpha)$ is nonincreasing.
Aggregate value can nevertheless rise if an increase in $\eta p(\alpha)$ offsets
the decline in $W(\alpha)$.

The selected mean can also rise despite increasing selection loss, because
\begin{equation}\label{eq:W-att}
  m(\alpha)-\eta
  =
  \frac{W(\alpha)}{p(\alpha)},
  \qquad p(\alpha) > 0.
\end{equation}
For example, excluding units whose values lie just above $\eta$ but below the
previous selected mean raises that mean at little loss. Accounting for both entry
and exit, for
$0\leq\alpha_1<\alpha_2$ with $p(\alpha_i)>0$ for $i=1,2$,
\eqref{eq:W-att} gives $m(\alpha_2)>m(\alpha_1)$ if and only if
\begin{equation}\label{eq:m21}
  \big[m(\alpha_1) - \eta \big] [p(\alpha_1)-p(\alpha_2)] 
  > W(\alpha_1)-W(\alpha_2).
\end{equation}
When $W(\alpha_1)>0$, this condition requires a larger proportional decline in the
selection rate than in threshold-adjusted value. These comparisons require no
independence or distributional restrictions.

\subsection{Gaussian selected means}
\label{sec:ft-gaussian}

We now specialize to joint normality to characterize how the selected
mean responds to increasing noise.

\begin{assumption}\label{ass:normal}
  The random vector $(V,\eps)$ follows a bivariate normal distribution with mean
  $(\mu,0)$ and correlation $\rho\in(-1,1)$. The standard deviations of $V$
  and $\eps$ are $\sigma>0$ and $\sigeps>0$, respectively.
\end{assumption}

Under this assumption, $W(\alpha)$ is strictly decreasing for every
$\rho\in(-1,1)$, whereas $m(\alpha)$ need not be.
Let $r_\alpha := \alpha\sigma_\eps/\sigma$ denote \emph{relative noise scale}, the
ratio of the standard deviations of $\alpha\eps$ and $V$. For $\alpha\geq0$, define
\begin{equation}\label{eq:para}
  s(\alpha)
  := \text{Std}(V+\alpha\eps)
  =  \sigma\sqrt{1+2\rho r_\alpha + r_\alpha^2} > 0,
  \qquad
  q(\alpha)
  := \Corr(V,V+\alpha\eps)
  =  \frac{\sigma(1+\rho r_\alpha)}{s(\alpha)}.
\end{equation}
The assessment standard deviation $s(\alpha)$ increases with $\alpha$ whenever
$\rho>-r_\alpha$. For $\alpha>0$, the \emph{value--assessment correlation} $q(\alpha)$ strictly
decreases as $\alpha$ increases, since
\(
  q'(\alpha)
  =
  -s(\alpha)^{-3}\alpha\sigma\sigeps^2(1-\rho^2)<0.
\)
For $\alpha>0$, $q(\alpha) > 0$ if and only if $\rho>-1/r_\alpha$.

Let $\phi$ and $\Phi$ denote the standard normal density and cdf, respectively, and
define the lower-tail inverse Mills ratio $\lambda(z)=\phi(z)/\Phi(z)$ for
$z\in\RR$.

By Proposition~S.1, for $\alpha\geq0$, the selection rate and selected mean are
\begin{align}
  p(\alpha)
  &=
  \Phi\left(\frac{\mu-\eta}{s(\alpha)}\right),
  \label{eq:p-alp}\\
  m(\alpha)
  &=
  \mu+\sigma q(\alpha) (\lambda\circ\Phi^{-1})(p(\alpha)).
  \label{eq:gsm}
\end{align}
Thus $p(0) = \Phi((\mu-\eta)/\sigma)$ and $p(\alpha) \to 1/2$ as $\alpha\to\infty$.
Differentiating \eqref{eq:p-alp} gives
\begin{equation}\label{eq:p'}
  \sgn(p'(\alpha)) = -\sgn\big((\rho+r_\alpha)(p(0)-1/2)\big).
\end{equation}

Equation~\eqref{eq:gsm} separates two channels through which noise affects the
selected mean. Holding the selection rate fixed, a decline in $q(\alpha)$ lowers the
selected mean. Holding $q(\alpha)>0$ fixed, a decline in $p(\alpha)$ raises the
selected mean because $\lambda\circ\Phi^{-1}$ is strictly decreasing. The following
proposition gives the exact condition for the selected mean to rise with noise.
To state the result, define
\[
  \psi(z):=z [z+\lambda(z)], \qquad z\in\RR.
\]
Lemma S.2 shows that $\psi$ is strictly increasing from $\RR$ onto $(-1,\infty)$.

\begin{proposition}\label{prop:gcs}
  Suppose Assumption~\ref{ass:normal} holds and $\eta(\alpha)=\eta\in\RR$ is fixed.
  For every $\rho\in(-1,1)$ and $\alpha>0$,
  \begin{equation}\label{eq:gms}
    \sgn\bigl(m'(\alpha)\bigr)
    =\sgn\!\left[
      (1+\rho r_\alpha)(\rho+r_\alpha)\,
      \psi\!\left(\frac{\mu-\eta}{s(\alpha)}\right)
      -r_\alpha(1-\rho^2)
    \right].
  \end{equation}
  If, in addition, $\rho>-\min\{r_\alpha,1/r_\alpha\}$, then $m'(\alpha)>0$ if and only if 
  \begin{equation}\label{eq:grc}
    p(\alpha)
    >
    (\Phi\circ \psi^{-1})
    \left(\frac{r_\alpha(1-\rho^2)}{(1+\rho r_\alpha)(r_\alpha+\rho)}\right).
  \end{equation}
\end{proposition}

The next proposition gives the complete global classification for $\rho\geq0$,
using only the noise-free rate $p(0)$ and the value--error correlation.
We say $m(\alpha)$ is \emph{single-peaked} if it strictly
increases to a unique maximum at $\alpha^*>0$ and strictly decreases thereafter.
Define
\begin{equation}\label{eq:cutoff}
  p^\dagger:=(\Phi\circ\psi^{-1})(1),
  \qquad
  \bar p_\rho:=\Phi\!\left(\frac{1-\rho^2}{\rho\sqrt{2/\pi}}\right),
  \quad 0<\rho<1.
\end{equation}
In particular, $p^\dagger\approx 0.7995$ and $\bar p_\rho>1/2$ for $0<\rho<1$.

\begin{figure}[!t]
  \centering
  \includegraphics[width=\textwidth]{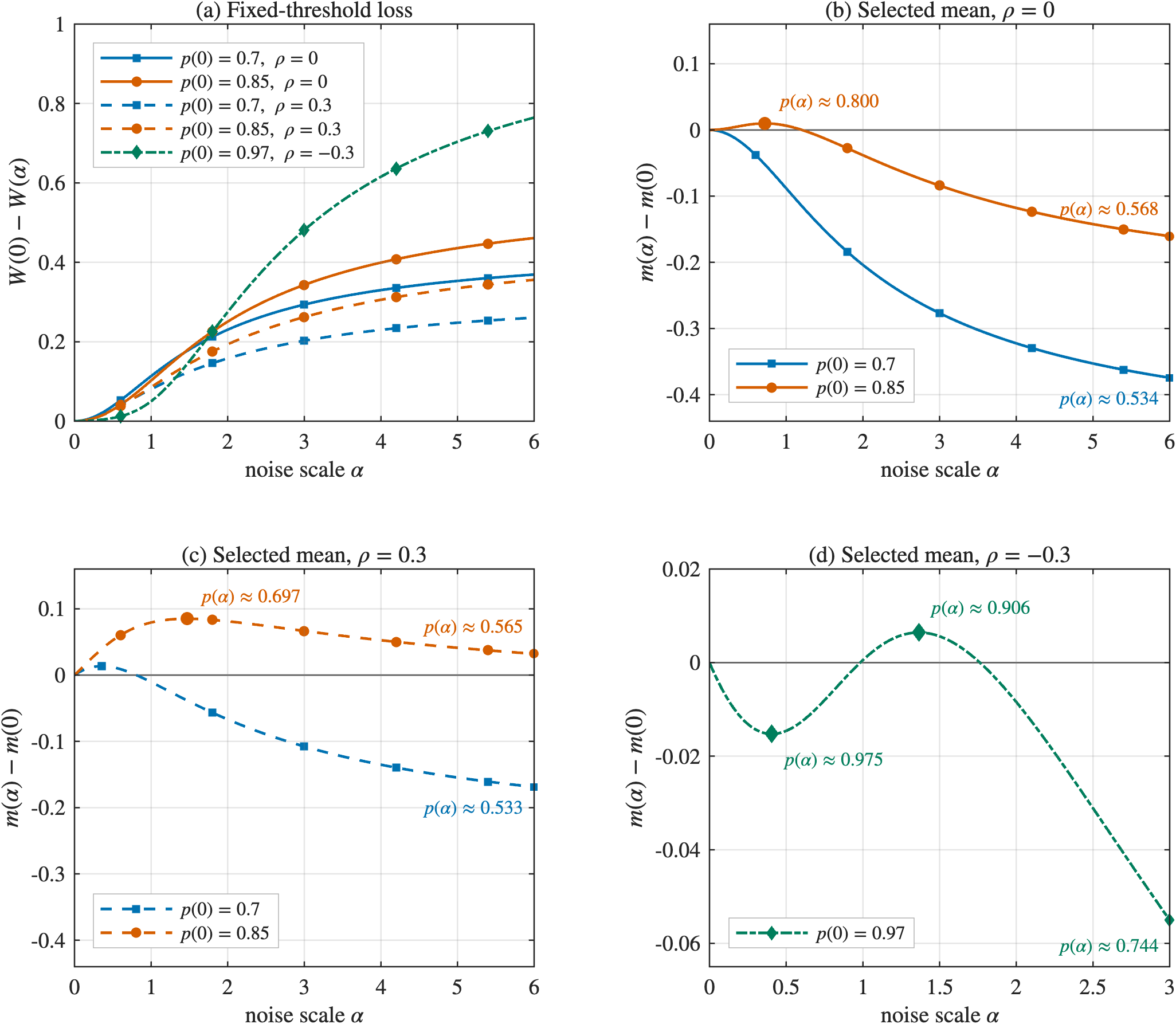}
  \caption{Fixed-threshold Gaussian selection, with $\mu=0$ and $\sigma=\sigeps=1$.
  (a) Loss for the paths in (b)--(d). (b)--(d) Selected-mean changes under zero,
  positive, and negative correlation. Legends give the noise-free selection rate
  $p(0)$; annotations give $p(\alpha)$ at marked extrema and endpoints.
  In (d), the selected mean initially falls, then recovers above $m(0)$ before
  falling below it again.}
  \label{fig:results}
\end{figure}

\begin{proposition}\label{prop:shape}
  Suppose Assumption~\ref{ass:normal} holds with $\rho\geq0$, and
  $\eta(\alpha)=\eta\in\RR$. As $\alpha$ ranges over $[0,\infty)$:
  \begin{enumerate}
    \item[(i)] 
      If $\rho=0$, $m(\alpha)$ is strictly decreasing when $p(0)\leq p^\dagger$
      and single-peaked when $p(0)>p^\dagger$.
    \item[(ii)] 
      If $\rho>0$, $m(\alpha)$ is strictly decreasing when $p(0)\leq1/2$,
      single-peaked when $1/2<p(0)<\bar p_\rho$, and strictly increasing when
      $p(0)\geq\bar p_\rho$.
  \end{enumerate}
\end{proposition}

In Figure~\ref{fig:results}, panel (a) shows loss increasing along every path,
while panels (b)--(d) show that the selected mean can exceed its noise-free level.
Under independence, panel (b) contrasts the decreasing path at $p(0)=0.7$ with
the single-peaked path at $p(0)=0.85$. In both paths, the threshold lies below
mean true value, so more units leave selection than enter it as noise increases.
For any $p(0)>p^\dagger$, selection contracts
throughout, and the maximum occurs when $p(\alpha)=p^\dagger\approx0.7995$.
Before this maximum, $p(\alpha)$ falls proportionately faster than $W(\alpha)$,
allowing the selected mean to rise despite increasing loss.

With positive correlation, even the path starting at $p(0)=0.7$ initially rises
(panel (c)). The selection rate at the maximum now depends on $p(0)$ and $\rho$;
it is approximately $0.697$ for the illustrated $p(0)=0.85$ path.
For fixed $1/2<p(0)\leq p^\dagger$, the maximum gain above $m(0)$ vanishes as
$\rho\downarrow0$, consistent with the decreasing path under independence.

For $\rho<0$, Proposition~\ref{prop:gcs} implies that the selected mean initially
rises when $p(0)<1/2$ and falls when $p(0)>1/2$; it decreases throughout when
$p(0)=1/2$. At low selection rates, negative correlation initially lowers the
assessment standard deviation and contracts selection. Since $q'(0)=0$, this
contraction accounts for the initial rise in the selected mean.

An initial decline can also be reversed, as panel (d) illustrates with
$\rho=-0.3$ and $p(0)=0.97$. The selected mean falls to a local minimum at
$p(\alpha)\approx0.975$, then recovers above $m(0)$ and reaches a local maximum
at $p(\alpha)\approx0.906$, before declining again. This sequence has two turning
points, whereas Proposition~\ref{prop:shape} allows at most one.
The Supplementary Material gives proofs and further local selection-rate
conditions for $\rho<0$.

\section{Subgroup ATT in treatment selection}\label{sec:subgroup-att}

For treatment selection, let $Y(1)$ and $Y(0)$ denote potential outcomes under
treatment and no treatment, respectively, and set $V=Y(1)-Y(0)$ as the treatment
gain. Then, $m(\alpha)$ is the ATT, $p(\alpha)$ is the treatment rate, and
$A(\alpha)=\E[(Y(1)-Y(0))D_\alpha]=p(\alpha)\mathrm{ATT}(\alpha)$ is the average
treatment gain generated by the selection rule, relative to universal
nonparticipation. Suppose two subgroups, indexed by $g\in\{1,2\}$, share the same
joint distribution of $(V,\eps)$ but have error scales $0\leq\alpha_1<\alpha_2$.
Their ATTs can then differ through selection alone, despite identical treatment-gain
distributions and hence identical average treatment effects (ATEs) in the two
subgroups.

If the same treatment rate is imposed within each subgroup,
Proposition~\ref{prop:fosd} gives a weakly lower ATT in group 2.
Equal subgroup treatment rates therefore do not eliminate differences in the
composition of participants. Fixing only the overall population treatment rate
does not yield this ordering.

At a common fixed participation cost $\eta$, excluded from treatment gains, group 2
has weakly lower net surplus per group member by Proposition~\ref{prop:ftv}.
Its ATT can nevertheless be higher: under independent Gaussian assessment errors with
$p(0)>p^\dagger$, this occurs when both noise scales lie before the maximum of
the selected mean (Figure~\ref{fig:results}(b)). In this case, reducing noise
expands treatment and raises net surplus even though ATT falls. Using ATT alone
to assess a change in decision procedures can therefore favor a rule with lower net surplus.
An observed ATT gap alone does not identify its source: it may reflect differences in
underlying gains, decision noise, or thresholds.

\section{Conclusion}

Reducing noise cannot lower the selected mean at a fixed selection rate, but can
lower it at a fixed threshold despite reducing selection loss. In treatment
selection, noise can generate subgroup ATT differences despite identical population
gain distributions. These comparisons scale a fixed assessment error while preserving
its joint distribution with true value. A higher selected mean therefore need not
indicate lower selection loss.

\begingroup
\setstretch{1.1}
\setlength{\bibsep}{0em}
\bibliographystyle{apalike}
\bibliography{references}
\endgroup

% Supplementary material: counters restart with distinct hyperlink destinations.
\clearpage
\setcounter{section}{0}
\setcounter{equation}{0}
\setcounter{figure}{0}
\setcounter{table}{0}
\setcounter{lemma}{0}
\setcounter{corollary}{0}
\setcounter{proposition}{0}
\renewcommand{\theequation}{S.\arabic{equation}}
\renewcommand{\thefigure}{S.\arabic{figure}}
\renewcommand{\thetable}{S.\arabic{table}}
\renewcommand{\thelemma}{S.\arabic{lemma}}
\renewcommand{\thecorollary}{S.\arabic{corollary}}
\renewcommand{\theproposition}{S.\arabic{proposition}}
\renewcommand{\theHsection}{supp.\arabic{section}}
\renewcommand{\theHequation}{supp.\arabic{equation}}
\renewcommand{\theHfigure}{supp.\arabic{figure}}
\renewcommand{\theHtable}{supp.\arabic{table}}
\renewcommand{\theHlemma}{supp.\arabic{lemma}}
\renewcommand{\theHcorollary}{supp.\arabic{corollary}}
\renewcommand{\theHproposition}{supp.\arabic{proposition}}

\begin{center}
  {\Large\bfseries Supplementary Material\par}
  \medskip
  {\large\itshape Decision Noise in Threshold Selection\par}
\end{center}
\medskip

In this supplementary material, Section~\ref{sec:s-distribution-free} proves the
distribution-free comparisons in Propositions~\ref{prop:fosd} and~\ref{prop:ftv} of
the main paper. Section~\ref{sec:s-gaussian} collects the Gaussian selection
formulas, proves Propositions~\ref{prop:gcs} and~\ref{prop:shape}, and studies how
the cutoffs and the shape of $m(\alpha)$ vary with correlation. The local response is
characterized for every correlation; the complete global shape classification applies
to nonnegative correlation, with additional possibilities established under negative
correlation. Section~\ref{sec:s-subgroup-att} applies the results to subgroup ATT
comparisons in treatment selection.

We use the main paper's notation $m(\alpha)=\E(V\mid D_\alpha=1)$ for the
selected mean, defined when $\E|V|<\infty$ and $p(\alpha)>0$. The selection
rule, and hence $m$, depends on whether the rate or threshold is held fixed.

\section{Distribution-free comparisons}\label{sec:s-distribution-free}

\subsection{Selection at a fixed rate}

Let $D_1$ and $D_2$ be Bernoulli random variables jointly defined with a real-valued
random variable $V$ whose probability measure on $\RR$ is denoted by $F_{\Vs}$.
For $j=1,2$, define the finite measure $\mu_j$ by 
\(
  \mu_j(A) := \PP(V\in A, D_j=1)
\)
for any Borel set $A\subseteq\RR$. Since $0\leq \mu_j(A) \leq \PP(V\in A) =
F_{\Vs}(A)$, $\mu_j$ is absolutely continuous with respect to $F_{\Vs}$, and there exists
a measurable function $\pi_j:\RR\to[0,1]$, unique $F_{\Vs}$-almost everywhere, such that
\(
  \mu_j(A) = \int_A \pi_j(v)F_{\Vs}(\diff v)
\)
for every Borel set $A\subseteq\RR$. The defining property of \(\pi_j\) implies that
\(
  \E[1\{V\in A\}\pi_j(V)] = \E[1\{V\in A\}D_j]
\)
for every Borel set $A\subseteq\RR$. Hence, $\pi_j(V)=\E(D_j\mid V)$, $\PP$-almost surely.
In particular, $\pi_j$ can be interpreted as the conditional choice probability, that
is, $\pi_j(v) = \PP(D_j=1 \mid V = v)$ for $F_{\Vs}$-almost every $v\in\RR$.

Given these definitions, we have the following lemma, in which ``a.e.'' stands for
``almost every''.

\begin{lemma}\label{lem:crossing}
  Suppose that $\PP(D_1=1)=\PP(D_2=1)=p\in(0,1)$ and that, for some $v_c\in\mathbb
  R$, 
  $\pi_1(v)\leq\pi_2(v)$ for $F_{\Vs}$-a.e. $v\in(-\infty,v_c)$, and
  $\pi_1(v)\geq\pi_2(v)$ for $F_{\Vs}$-a.e. $v\in(v_c, \infty)$.
  Then
  \begin{equation} \label{eq:cdf-ieq}
    \PP(V\leq v\mid D_1=1)
    \leq
    \PP(V\leq v\mid D_2=1)
    \qquad\text{for every }v\in\mathbb R.
  \end{equation}
  Equivalently, $V\mid D_1=1$ first-order stochastically dominates
  $V\mid D_2=1$. If
  \begin{equation}\label{eq:scc}
   \PP\{\pi_1(V)\neq\pi_2(V)\}>0
  \end{equation}
  holds, the inequality in \eqref{eq:cdf-ieq} is strict for some $v\in\RR$.
  If \eqref{eq:scc} holds and, in addition, $\E|V|<\infty$, then
  \begin{equation} \label{eq:E-ieq}
    \E(V\mid D_1=1)>\E(V\mid D_2=1).
  \end{equation}
\end{lemma}

\begin{proof}[Proof of Lemma~\ref{lem:crossing}]
  For any $v\in\mathbb R$, define $G_j(v)=\PP(V\leq v\mid D_j=1)$ for $j=1,2$. To
  show \eqref{eq:cdf-ieq} it suffices to show that $G_1(v)-G_2(v) \leq 0$ for every
  $v\in\RR$. Using $\PP(D_j=1)=p\in(0,1)$ and applying the law of iterated expectations,
  \begin{align*}
    G_j(v) 
    =
    \frac1p\PP(V\leq v, D_j=1)
    =
    \frac1p \E(1\{V\leq v\} D_j) 
    &=
    \frac1p \E \big[ 1\{V\leq v\} \E( D_j\mid V) \big]
    \\
    &=
    \frac1p\int_{(-\infty,\,v]} \pi_j(x) F_{\Vs}(\diff x).
  \end{align*}
  Thus
  \begin{equation}\label{eq:cdf-device}
    G_1(v)-G_2(v)
    =
    \frac{1}{p}\int_{(-\infty,\,v]} [\pi_1(x) - \pi_2(x)]\,F_{\Vs}(\diff x),
    \qquad 
    v\in\RR.
  \end{equation}
  Since $\pi_1(x)-\pi_2(x)\leq0$ for $F_{\Vs}$-a.e. $x\in(-\infty,v_c)$, we have
  $G_1(v)-G_2(v)\leq0$ for every $v<v_c$.
  Moreover, by \eqref{eq:cdf-device} and 
  \begin{equation}\label{eq:zr}
    \int_{\RR} [\pi_1(x)-\pi_2(x)] F_{\Vs}(\diff x)
    = 
    \E[\pi_1(V)] - \E[\pi_2(V)]
    = 
    \E(D_1) - \E(D_2) = p-p=0,
  \end{equation}
  we have
  \begin{equation}\label{eq:cdf-device-2}
    G_1(v) - G_2(v) 
    =
    - \frac1p\int_{(v,\infty)} [\pi_1(x)-\pi_2(x)] F_{\Vs}(\diff x),
    \qquad
    v\in\RR.
  \end{equation}
  Since $\pi_1(x)-\pi_2(x)\geq0$ for $F_{\Vs}$-a.e. $x\in(v_c,\infty)$, it then follows that
  $G_1(v)-G_2(v)\leq0$ for every $v\geq v_c$. Hence \eqref{eq:cdf-ieq} holds.

  For the strict inequality, let
  \(
    \Delta(v):=\pi_1(v)-\pi_2(v)
  \)
  and define
  \(
    A:=\{v\in\RR:v\neq v_c,\ \Delta(v)\neq0\}.
  \)
  We first show that \eqref{eq:scc} implies $F_{\Vs}(A)>0$. Suppose, to the
  contrary, that $F_{\Vs}(A)=0$. Then $\Delta(v)=0$ for $F_{\Vs}$-a.e. $v\in(-\infty,
  v_c)\cup (v_c,\infty)$, and hence \eqref{eq:zr} gives
  \(
    0
    =
    \int_{\RR}\Delta(v)F_{\Vs}(\diff v)
    =
    \Delta(v_c)F_{\Vs}(\{v_c\}).
  \)
  Moreover, $F_{\Vs}(A) = 0$ also implies that
  \[
    \PP\{\Delta(V)\neq0\}
    = 
    F_{\Vs}(\{v\in\RR: \Delta(v)\neq0\})
    =
    1\{\Delta(v_c)\neq0\}F_{\Vs}(\{v_c\})
    =
    \begin{cases}
      F_{\Vs}(\{v_c\}), & \text{if } \Delta(v_c) \neq 0,\\
      0, & \text{if } \Delta(v_c) = 0.
    \end{cases}
  \]
  Since $\Delta(v_c)F_{\Vs}(\{v_c\})=0$ implies that $F_{\Vs}(\{v_c\}) = 0$ if
  $\Delta(v_c)\neq0$, it follows that $\PP\{\Delta(V)\neq0\}=0$, which contradicts
  \eqref{eq:scc}. Therefore, $F_{\Vs}(A)>0$.

  Now define 
  \begin{align*}
    A_- &:=A\cap(-\infty,v_c) = \{v\in(-\infty,v_c):\Delta(v)\neq 0\},
    \\
    A_+ &:=A\cap(v_c,\infty) = \{v\in(v_c,\infty):\Delta(v)\neq 0\}
  \end{align*}
  Since $F_{\Vs}(A)>0$, at least one of $F_{\Vs}(A_-)$ and $F_{\Vs}(A_+)$ is
  positive. If $F_{\Vs}(A_+)>0$, then we note
  \[
    \int_{(v_c,\infty)}\Delta(x)F_{\Vs}(\diff x) 
    =
    \int_{A_+}\Delta(x)F_{\Vs}(\diff x) 
    >0.
  \]
  where the inequality holds because $\Delta(v)>0$ for $F_{\Vs}$-a.e. $v\in A_+$ and
  $F_{\Vs}(A_+)>0$. Then, by \eqref{eq:cdf-device-2}, we have $G_1(v_c)-G_2(v_c)<0$.
  If instead $F_{\Vs}(A_-)>0$, take any sequence $v_n\uparrow v_c$ with
  $v_n<v_c$. Since
  \(
    A_-=
    \bigcup_{n=1}^{\infty}
    \bigl(A_-\cap(-\infty,v_n]\bigr),
  \)
  there exists some $n$ such that
  \(
    F_{\Vs}\bigl(A_-\cap(-\infty,v_n]\bigr)>0.
  \)
  Because $\Delta(v)<0$ for $F_{\Vs}$-a.e. $v\in A_-$, we then have
  \[
    \int_{(-\infty,v_n]}\Delta(x)F_{\Vs}(\diff x) 
    = 
    \int_{A_-\cap(-\infty,v_n]} \Delta(x) F_{\Vs}(\diff x)
    <0.
  \]
  By \eqref{eq:cdf-device}, it follows that
  \(
    G_1(v_n)-G_2(v_n)<0.
  \)
  Therefore, whenever \eqref{eq:scc} holds, the inequality in \eqref{eq:cdf-ieq} is
  strict for some $v\in\RR$.

  Lastly, suppose that \eqref{eq:scc} holds and $\E|V|<\infty$. Then,
  \(
    \E(|V|\mid D_j=1) = p^{-1} \E(|V|D_j) \leq p^{-1}\E|V| <\infty,
  \)
  for $j=1,2$. Together with \eqref{eq:zr}, this gives
  \[
    \E(V\mid D_1=1)-\E(V\mid D_2=1)
    =
    \frac{1}{p}
    \int_{\RR}(v-v_c)\Delta(v)F_{\Vs}(\diff v).
  \]
  By the assumptions in the lemma, $(v-v_c)\Delta(v)\geq0$ for $F_{\Vs}$-a.s.
  $v\in\RR$, and the inequality is strict for $F_{\Vs}$-a.e. $v\in A$.
  Since, as shown above, \eqref{eq:scc} implies $F_{\Vs}(A)>0$, it follows
  that
  \[
    \int_{\RR}(v-v_c)\Delta(v)F_{\Vs}(\diff v)
    \geq
    \int_A (v-v_c)\Delta(v)F_{\Vs}(\diff v)
    >0.
  \]
  Therefore, $\E(V\mid D_1=1)-\E(V\mid D_2=1) > 0$, which proves \eqref{eq:E-ieq},
  and hence the whole lemma.
\end{proof}

\medskip

Using Lemma~\ref{lem:crossing}, we prove Proposition~\ref{prop:fosd} in the paper.
In the fixed-rate setting, to allow the noisy assessment $V+\alpha\eps$ to have an
atom at the adjustable threshold $\eta(\alpha)$, we specify the following randomized
tie-breaking selection rule (see Footnote~1 in the paper). For $\alpha \geq0$,
\begin{equation}\label{eq:Dj}
  D_{\alpha}
  :=
  1\{V+\alpha\eps > \eta(\alpha)\} 
  +
  1\{V+\alpha\eps = \eta(\alpha)\}
  1\{U \leq \gamma_{\alpha}\},
  \qquad 
  \PP(D_\alpha=1) = p \in (0,1), 
\end{equation}
where $U$ is uniform on $(0,1)$ and independent of $(V,\eps)$. If
$\PP(V+\alpha\eps=\eta(\alpha))>0$, the target rate $p = \PP(D_\alpha=1)$ pins down
\begin{equation}\label{eq:gamma}
  \gamma_\alpha 
  =
  \frac{p-\PP(V+\alpha\eps>\eta(\alpha))}
  {\PP(V+\alpha\eps=\eta(\alpha))}.
\end{equation}

\begin{proof}[Proof of Proposition~\ref{prop:fosd}]
  Fix $0\leq\alpha_1<\alpha_2$. For each $j=1,2$, $D_{\alpha_j}$ is defined by
  \eqref{eq:Dj} and \eqref{eq:gamma}. To prove the first-order stochastic dominance
  in the proposition, we must show
  \begin{equation}\label{eq:s-fosd}
    \PP(V\leq v\mid D_{\alpha_1}=1)
    \leq 
    \PP(V\leq v\mid D_{\alpha_2}=1)
    \qquad\text{for every }v\in\mathbb R.
  \end{equation}

  Let $F_{\Vs}$ denote the probability measure of $V$ on $\RR$.
  For $j=1,2$, let $\pi_j:\RR\to[0,1]$ be the measurable function, unique
  $F_{\Vs}$-a.e., such that 
  \[
    \PP(V\in A, D_{\alpha_j}=1) = \int_A \pi_j(v)F_{\Vs}(\diff v)
    \quad
    \text{for every Borel set } A\subseteq\RR.
  \]
  By Lemma~\ref{lem:crossing}, to show \eqref{eq:s-fosd} it suffices to show
  that there exists some $v_c\in \RR$ such that
  \begin{equation}\label{eq:cr}
    \begin{aligned}
      \pi_1(v) &\leq \pi_2(v)
      &&\quad\text{for $F_{\Vs}$-a.e. } v < v_c,
      \\
      \pi_1(v) &\geq \pi_2(v)
      &&\quad\text{for $F_{\Vs}$-a.e. } v > v_c.
    \end{aligned}
  \end{equation}

  First, consider $\alpha_1=0<\alpha_2$. In this case, 
  \(
    D_{\alpha_1} = 1\{V>\eta(0)\} + 1\{V=\eta(0)\} 1\{U\leq \gamma_0\} 
  \)
  for some $\eta(0)$ and $\gamma_0$ such that $\PP(D_{\alpha_1}=1) = p$.
  By the definition of $\pi_1$, we have
  \[
    \pi_1(v) = 
    \begin{cases}
      0 & \text{for $F_{\Vs}$-a.e. } v<\eta(0) \\
      1 & \text{for $F_{\Vs}$-a.e. } v>\eta(0).
    \end{cases}
  \]
  Since $\pi_2(v)\in[0,1]$ for $v\in\RR$, it follows immediately that
  $\pi_2(v)\geq \pi_1(v)$ for $F_{\Vs}$-a.e. $v<\eta(0)$ and 
  $\pi_2(v)\leq \pi_1(v)$ for $F_{\Vs}$-a.e. $v>\eta(0)$. Setting $v_c =
  \eta(0)$ then proves \eqref{eq:cr} and hence \eqref{eq:s-fosd}.
  
  Next, consider $0<\alpha_1<\alpha_2$. For $j=1,2$, it follows from the definition
  of $D_{\alpha_j}$ in \eqref{eq:Dj} and the definition of $\pi_j$ that 
  \begin{equation}\label{eq:piv}
   \pi_j(v)
   =
   \PP \left\{\left. \eps> \frac{\eta(\alpha_j)-v}{\alpha_j} \right| V=v \right\}
   +
   \gamma_{\alpha_j} 
   \PP \left\{\left. \eps= \frac{\eta(\alpha_j)-v}{\alpha_j} \right| V=v \right\},
   \quad 
   F_{\Vs}\mbox{-a.e. } v.
  \end{equation}
  Let
  \(
    v_c
    =
    (\alpha_2\eta(\alpha_1) - \alpha_1\eta(\alpha_2))/(\alpha_2-\alpha_1).
  \)
  It can be easily shown that 
  \(
    (\eta(\alpha_1)-v)/\alpha_1 > (\eta(\alpha_2)-v)/\alpha_2 
  \)
  if $v < v_c$ and 
  \(
    (\eta(\alpha_1)-v)/\alpha_1 < (\eta(\alpha_2)-v)/\alpha_2 
  \)
  if $v > v_c$. 
  By \eqref{eq:piv} and $\gamma_{\alpha_j} \in [0,1]$, we can show that, for
  $F_{\Vs}$-a.e. $v<v_c$,
  \[
    \pi_1(v)
    \leq
    \PP \left\{\left. \eps \geq \frac{\eta(\alpha_1)-v}{\alpha_1} \right| V=v \right\}
    \leq
    \PP \left\{\left. \eps > \frac{\eta(\alpha_2)-v}{\alpha_2} \right| V=v \right\}
    \leq
    \pi_2(v),
  \]
  while for $F_{\Vs}$-a.e. $v>v_c$,
  \[
    \pi_1(v)
    \geq
    \PP \left\{\left. \eps > \frac{\eta(\alpha_1)-v}{\alpha_1} \right| V=v \right\}
    \geq
    \PP \left\{\left. \eps \geq \frac{\eta(\alpha_2)-v}{\alpha_2} \right| V=v \right\}
    \geq
    \pi_2(v).
  \]
  This shows \eqref{eq:cr} and hence \eqref{eq:s-fosd}.
  
  Lastly, if $\E|V|<\infty$, then $\E(|V|\mid D_\alpha=1)\leq\E|V|/p<\infty$ for
  every $\alpha$. The first-order stochastic-dominance ordering implies
  $m(\alpha_1)\geq m(\alpha_2)$.
  This proves the conditional mean statement in the proposition.
\end{proof}

\subsection{Selection at a fixed threshold}

Under the fixed threshold, $\eta(\alpha) = \eta$, and the threshold selection rule is
\[
  D_\alpha=1\{V+\alpha\eps>\eta\}, \qquad \alpha\geq0.
\]
The noise-free rule is $D_0=1\{V>\eta\}$.

\begin{proof}[Proof of Proposition~\ref{prop:ftv}]
  Since $|(V-\eta)D_\alpha|\leq |V-\eta|$, the assumption $\E|V-\eta|<\infty$ ensures
  that all expectations below are finite. To prove \eqref{eq:threshold-loss}, note
  that $W(0)-W(\alpha) = \E[(V-\eta)(D_0-D_\alpha)]$, where
  \[
    (V-\eta)(D_0-D_\alpha)
    =
    \begin{cases}
      V-\eta, & D_0=1, D_\alpha =0;
      \\
      \eta-V, & D_0=0, D_\alpha =1;
      \\
      0, & D_0 = D_\alpha.
    \end{cases}
  \]
  Since $V-\eta = |V-\eta|$ if $D_0=1$ and $\eta-V = |V-\eta|$ if $D_0=0$, we have 
  \[
    (V-\eta)(D_0-D_\alpha) = |V-\eta| 1\{D_\alpha\neq D_0\},
    \qquad
    \alpha \geq 0.
  \]
  Taking expectations yields \eqref{eq:threshold-loss}.

  Next, we prove \eqref{eq:threshold-loss-2}. Fix $0\leq \alpha_1 < \alpha_2$.  We
  have $W(\alpha_1)-W(\alpha_2) = \E[(V-\eta)(D_{\alpha_1}-D_{\alpha_2})]$. Since
  \[
    D_{\alpha_1} = 
    \begin{cases}
      1, & V+\alpha_1\eps > \eta,
      \\
      0, & V+\alpha_1\eps \leq \eta,
    \end{cases}
    \qquad
    D_{\alpha_2} = 
    \begin{cases}
      1, & V+\alpha_2\eps > \eta,
      \\
      0, & V+\alpha_2\eps \leq \eta.
    \end{cases}
  \]
  Hence
  \[
    D_{\alpha_1} - D_{\alpha_2}
    =
    \begin{cases}
      1, & -\alpha_1\eps < V-\eta \leq -\alpha_2\eps 
      \ \text{ or, equivalently, }\ (D_{\alpha_1}, D_{\alpha_2}) = (1,0),
      \\
      -1, & \alpha_1\eps\leq \eta-V < \alpha_2\eps
      \ \text{ or, equivalently, }\ (D_{\alpha_1}, D_{\alpha_2}) = (0,1),
      \\
      0, & \text{otherwise}
    \end{cases}
  \]
  Since $0\leq\alpha_1<\alpha_2$, we have
  \begin{align*}
    \{(D_{\alpha_1},D_{\alpha_2}) = (1,0)\}
    =
    \{-\alpha_1\eps < V-\eta \leq -\alpha_2\eps\}
    &=
    \{-\alpha_1\eps < V-\eta \leq -\alpha_2\eps\}\cap\{\eps<0\}=:E_1,
    \\
    \{(D_{\alpha_1},D_{\alpha_2}) = (0,1)\}
    =
    \{\alpha_1\eps\leq \eta-V < \alpha_2\eps\}
    &=
    \{\alpha_1\eps\leq \eta-V < \alpha_2\eps\}\cap\{\eps>0\}=:E_2.
  \end{align*}
  Clearly, $V-\eta>0$, and hence $D_0=1$, on $E_1$, while $\eta-V\geq0$, and hence
  $D_0=0$, on $E_2$. Thus
  \begin{align*}
    E_1 = 
    \{(D_{\alpha_1},D_{\alpha_2}) = (1,0)\} 
    &=
    \{(D_{\alpha_1},D_{\alpha_2}) = (1,0), \, D_0 = 1\},
    \\
    E_2 = 
    \{(D_{\alpha_1},D_{\alpha_2}) = (0,1)\} 
    &=
    \{(D_{\alpha_1},D_{\alpha_2}) = (0,1), \, D_0 = 0\},
  \end{align*}
  and
  \[
    (V-\eta)(D_{\alpha_1} - D_{\alpha_2})
    =
    \begin{cases}
      V-\eta>0, &  \text{if $E_1$ happens},
      \\
      \eta-V\geq0, &  \text{if $E_2$ happens},
      \\
      0, & \text{otherwise}.
    \end{cases}
  \]
  Since $E_1\cup E_2 = \{D_{\alpha_1} = D_0, \, D_{\alpha_2} \neq D_0\}$, we have 
  \[
    (V-\eta)(D_{\alpha_1} - D_{\alpha_2})
    =
    |V-\eta|\, 1\{D_{\alpha_1} = D_0, \, D_{\alpha_2} \neq D_0\}.
  \]
  Taking expectations then yields \eqref{eq:threshold-loss-2}. This completes the
  proof of Proposition~\ref{prop:ftv}.
\end{proof}

\section{Gaussian selection}\label{sec:s-gaussian}

In this section, we derive the results under joint normality of $(V,\eps)$ imposed by
Assumption~\ref{ass:normal} in the paper. Specifically, $(V,\eps)$ follows a
bivariate normal distribution with mean vector $(\mu,0)$ and correlation
$\rho\in(-1,1)$. The standard deviations of $V$ and $\eps$ are denoted by $\sigma>0$
and $\sigma_\eps>0$, respectively.

Let $r_\alpha := \alpha\sigma_\eps/\sigma$ denote \emph{relative noise scale}, the
standard deviation of decision noise $\alpha\eps$ relative to that of true value $V$.
Clearly, $r_0=0$ and $r_\alpha > 0$ for $\alpha > 0$. With $\sigma$ and $\sigma_\eps$
fixed, increasing $\alpha$ is equivalent to increasing $r_\alpha$. For $\alpha\geq0$,
define
\begin{equation}\label{eq:param}
  s(\alpha)
  := \text{Std}(V+\alpha\eps),
  \qquad
  q(\alpha)
  := \Corr(V,V+\alpha\eps)
\end{equation}
to be the standard deviation of the assessment and the value-assessment correlation,
respectively. Under Assumption~\ref{ass:normal},
\begin{equation}\label{eq:param2}
  s(\alpha)
  =
  \sigma\sqrt{1+2\rho r_\alpha + r_\alpha^2} > 0,
  \qquad
  q(\alpha)
  =  
  \frac{\sigma(1+\rho r_\alpha)}{s(\alpha)},
\end{equation}
as stated in equation~\eqref{eq:para} of Section~\ref{sec:ft-gaussian}.
In particular, $s(0) = \sigma$, and 
\begin{equation}\label{eq:ds}
  s'(\alpha)
  =
  \frac{\sigma\sigeps (r_\alpha + \rho)}{s(\alpha)}, \qquad \alpha > 0.
\end{equation}
Hence $s(\alpha)$ increases with $\alpha$ whenever $\rho > - r_\alpha$. Moreover,
$q(0)=1$,
\begin{equation}\label{eq:q-lim}
  \lim_{\alpha\to\infty} q(\alpha)
  =
  \lim_{\alpha\to\infty}
  \frac{\sigma/(\alpha\sigeps)+\rho}{\sqrt{\{\sigma/(\alpha\sigeps)\}^2
  +1
  +2\rho\sigma/(\alpha\sigeps)}}
  =
  \rho,
\end{equation}
and
\begin{equation}\label{eq:q-der}
  q'(\alpha)
  =
  -\frac{\alpha\sigma\sigeps^2(1-\rho^2)}{s(\alpha)^3}<0,
  \qquad \alpha > 0.
\end{equation}
Hence $q(\alpha)$ decreases strictly from $1$ toward $\rho$ as $\alpha$ increases.
Higher decision noise unambiguously lowers the correlation between the true value and
the noisy assessment.

Let $\phi$ and $\Phi$ denote the standard normal density and cdf, and define the
lower-tail inverse Mills ratio
\[
  \lambda(z) = \frac{\phi(z)}{\Phi(z)}, \qquad z\in\RR.
\]
For every $z\in\RR$, let $Y_z$ have the conditional distribution of $Z$ given $Z>-z$,
where $Z$ is a standard normal random variable. By the mean and variance formula of
the truncated normal distribution,
\begin{equation}\label{eq:tn}
  \E(Y_z) = \lambda(z),
  \qquad 
  \Var(Y_z) = 1 - \lambda(z)[z+\lambda(z)],
  \qquad z\in\RR.
\end{equation}
By definition, $\E(Y_z) = \E(Z\mid Z>-z)>-z$. Hence we have $\lambda(z)>\max\{-z,0\}$.
Furthermore, since $\phi'(z) = -z\phi(z)$, it can be
easily verified that
\begin{equation}\label{eq:ld}
  \lambda'(z) = -\lambda(z) [z+\lambda(z)] < 0, \qquad z\in\RR.
\end{equation}
By \eqref{eq:tn} and \eqref{eq:ld}, $\Var(Y_z) = 1+\lambda'(z) > 0$, from which
it follows that $\lambda'(z) > -1$. Consequently, we have
\begin{equation}\label{eq:mv}
  -1<\lambda'(z)<0, \qquad z\in\RR.
\end{equation}

Lastly, let $\Phi_2(\cdot,\cdot;\varrho)$ and $\phi_2(\cdot,\cdot;\varrho)$ denote,
respectively, the cdf and pdf of the standard bivariate normal distribution with
correlation $\varrho\in(-1,1)$.

\subsection{Gaussian selection formulas}\label{sec:gf}

For $\alpha\geq0$, the threshold selection rule
and the corresponding selection rate are given by
\[
  D_\alpha = 1\{ V + \alpha\eps > \eta(\alpha)\},
  \qquad
  p(\alpha) = \PP(D_\alpha =1) = \PP(V+\alpha\eps>\eta(\alpha)).
\]
Here, $\eta(\alpha)$ is adjusted to maintain a fixed selection rate in the fixed-rate
setting, but is held fixed in the fixed-threshold setting. Correspondingly,
$p(\alpha)$ is fixed in the fixed-rate setting and varies with $\alpha$ in the
fixed-threshold setting. If $\alpha = 0$, the selection is noise-free, and we have 
\[
  D_0 = 1\{ V > \eta(0)\},
  \qquad
  p(0) = \PP(V > \eta(0)).
\]

The following proposition provides formulas for the selection rate, the distribution of 
value among selected units, and the selected mean $m(\alpha)$,
under joint normality of $V$ and $\eps$. The corresponding results for the fixed-rate
and fixed-threshold settings are obtained by holding $p(\alpha)$ and $\eta(\alpha)$
fixed, respectively.

\begin{proposition}\label{prop:s-gaussian}
  Suppose Assumption~\ref{ass:normal} holds, and let
  $\tau_\alpha:=(\mu-\eta(\alpha))/s(\alpha)$ with $s(\alpha)$ defined in
  \eqref{eq:param2}. For every $\alpha\geq0$, the selection rate is
  \begin{equation}\label{eq:pa}
    p(\alpha) = \Phi(\tau_\alpha).
  \end{equation}
  If $\alpha>0$, the cdf of the value among selected units and the selected mean are
  given by
  \begin{align}
    \PP(V\leq v\mid D_\alpha=1)
    &=
    \frac
    {\Phi\big((v-\mu)/\sigma\big)-\Phi_2\big( (v-\mu)/\sigma, -\tau_\alpha; q(\alpha)\big)}
    {\Phi(\tau_\alpha)},
    \qquad v\in\RR,
    \label{eq:gs-cdf}
    \\[1mm]
    m(\alpha)
    &=
    \mu+\sigma q(\alpha) \lambda(\tau_\alpha).
    \label{eq:gs-mean}
  \end{align}
  At $\alpha=0$, the selection is noise-free, and the cdf of the value among selected
  units and the selected mean are given by
  \begin{align}
    \PP(V\leq v\mid D_0=1)
    &=
    \begin{cases}
      0, & v\leq\eta(0),\\[2mm]
      \displaystyle
      \frac
      {\Phi\big((v-\mu)/\sigma\big)-\Phi(-\tau_0)}
      {\Phi(\tau_0)}, & v>\eta(0),
    \end{cases}
    \qquad v\in\RR, \label{eq:gs-cdf-zero}
    \\[1mm]
    m(0)
    &=
    \mu+\sigma\lambda(\tau_0). \label{eq:gs-mean-zero}
  \end{align}
\end{proposition}

\begin{proof}[Proof of Proposition~\ref{prop:s-gaussian}]
  Under Assumption~\ref{ass:normal}, $V+\alpha\eps \sim N(\mu,s(\alpha)^2)$ for
  $\alpha\geq0$, and hence
  \[
    p(\alpha)
    = \PP(V+\alpha\eps > \eta(\alpha))
    = \PP\left(\frac{V+\alpha\eps - \mu}{s(\alpha)} > -\tau_\alpha \right)
    = 1 - \Phi(-\tau_\alpha) 
    = \Phi(\tau_\alpha).
  \]
  This proves \eqref{eq:pa}. To show \eqref{eq:gs-cdf}, let $V_\circ :=
  (V-\mu)/\sigma$, $V_\circ^\alpha := (V+\alpha\eps-\mu)/s(\alpha)$, and $v_\circ :=
  (v-\mu)/\sigma$ for $v\in\RR$. Then,
  $\{V\leq v\} = \{V_\circ \leq v_\circ\}$, 
  $\{D_\alpha=1\} = \{V_\circ^\alpha > -\tau_\alpha\}$, and
  \begin{equation} \label{eq:Vv}
    \PP(V\leq v \mid D_\alpha = 1)
    =
    \frac
    {\PP(V_\circ\leq v_\circ, V_\circ^\alpha > -\tau_\alpha)}
    {\PP(V^\alpha_\circ>-\tau_\alpha)}.
  \end{equation}
  When $\alpha>0$, $q(\alpha)\in(-1,1)$, and, under Assumption~\ref{ass:normal},
  $(V_\circ, V_\circ^\alpha)$ is standard bivariate normal with correlation
  $q(\alpha)$. Consequently, $\PP(V^\alpha_\circ>-\tau_\alpha) = \Phi(\tau_\alpha)$ 
  and
  \[
    \PP(V_\circ\leq v_\circ, V_\circ^\alpha > -\tau_\alpha)
    =
    \PP(V_\circ\leq v_\circ)
    -
    \PP(V_\circ\leq v_\circ,\, V_\circ^\alpha \leq -\tau_\alpha)
    =
    \Phi(v_\circ)-\Phi_2(v_\circ,-\tau_\alpha;\, q(\alpha)).
  \]
  Together with \eqref{eq:Vv}, this proves \eqref{eq:gs-cdf}.

  To show \eqref{eq:gs-mean}, we first note that
  \[
    \E(V_\circ\mid D_\alpha = 1)
    =
    \E(V_\circ\mid V_\circ^\alpha > -\tau_\alpha)
    =
    \E \big[ \E(V_\circ\mid V_\circ^\alpha) \mid V_\circ^\alpha > -\tau_\alpha \big],
  \]
  where the last equality holds by the law of iterated expectations. 
  Since $(V_\circ,V_\circ^\alpha)$ is standard bivariate normal with correlation
  $q(\alpha)\in(-1,1)$, the conditioning property of the multivariate normal
  distribution implies that $\E(V_\circ\mid V_\circ^\alpha) = q(\alpha)
  V_\circ^\alpha$, and therefore,
  \begin{equation}\label{eq:v1}
    \E(V_\circ\mid D_\alpha = 1)
    =
    q(\alpha) \, \E (V_\circ^\alpha \mid V_\circ^\alpha > -\tau_\alpha).
  \end{equation}
  Since $V_\circ^\alpha$ is a standard normal random variable, it follows from
  \eqref{eq:tn} that
  \begin{equation}\label{eq:v2}
    \E(V_\circ^\alpha \mid V_\circ^\alpha > -\tau_\alpha)
    = 
    \lambda(\tau_\alpha).
  \end{equation}
  Equation~\eqref{eq:gs-mean} then follows immediately from \eqref{eq:v1},
  \eqref{eq:v2}, and $V = \mu + \sigma V_\circ$. 

  At $\alpha=0$, we have
  $\{V\leq v\} = \{V_\circ \leq (v-\mu)/\sigma\}$ and 
  $\{D_0=1\} = \{V_\circ > (\eta(0)-\mu)/\sigma\}$. 
  So, if $v\leq\eta(0)$, then $\{V\leq v\}\cap\{D_0=1\}=\varnothing$, and hence
  \(
    \PP(V\leq v\mid D_0=1)=0.
  \)
  If $v>\eta(0)$, then, with $v_\circ = (v-\mu)/\sigma$ and $\tau_0 =
  (\mu-\eta(0))/\sigma$, we have
  \[
    \PP(V\leq v\mid D_0=1)
    =
    \frac{\PP(-\tau_0<V_\circ\leq v_\circ)}{\PP(V_\circ>-\tau_0)}
    =
    \frac{\Phi(v_\circ)-\Phi(-\tau_0)}{1-\Phi(-\tau_0)}
    =
    \frac{\Phi(v_\circ)-\Phi(-\tau_0)}{\Phi(\tau_0)}.
  \]
  This proves \eqref{eq:gs-cdf-zero}. Finally, since $V_\circ$ is a standard normal
  random variable, applying \eqref{eq:tn} again yields
  \(
    m(0)
    =
    \mu+\sigma\,\E(V_\circ\mid V_\circ>-\tau_0)
    =
    \mu+\sigma\lambda(\tau_0),
  \)
  which proves \eqref{eq:gs-mean-zero}.
\end{proof}

\subsection{Selection at a fixed rate}\label{sec:gfr}

Equation~\eqref{eq:pa} in Proposition~\ref{prop:s-gaussian} shows that
\(
  p(\alpha) = \Phi\big((\mu-\eta(\alpha))/s(\alpha)\big)
\)
under the joint normality of $(V,\eps)$. Under the fixed-rate setting, the selection
rate $p(\alpha)=p\in(0,1)$ is fixed, so that the selection threshold $\eta(\alpha)$
adjusts and is determined by 
\begin{equation}\label{eq:eta-alp}
  \eta(\alpha) = \mu-s(\alpha)\,\Phi^{-1}(p), \qquad \alpha \geq 0.
\end{equation}
The following corollary to Proposition~\ref{prop:s-gaussian} illustrates the results
in Proposition~\ref{prop:fosd} under joint normality of $(V,\eps)$.

\begin{corollary} \label{cor:s-g-fr}
  Suppose Assumption~\ref{ass:normal} holds and the selection rate $\PP(D_\alpha=1) =
  p \in(0,1)$ is fixed. Then, for every $\alpha>0$ and $v\in\RR$,
  \begin{align}\label{eq:fp}
    \frac{\mathrm d}{\mathrm d\alpha} \PP(V\leq v\mid D_\alpha=1)
    >0
    \qquad\mbox{and}\qquad
    m'(\alpha)
    < 0.
  \end{align}
  Moreover,
  \begin{equation}\label{eq:fp-lim}
    m(\alpha) 
    \to
    \mu + \sigma\rho \lambda\big(\Phi^{-1}(p)\big),
    \qquad \text{as } \alpha\to\infty.
  \end{equation}
\end{corollary}

\begin{proof}[Proof of Corollary~\ref{cor:s-g-fr}]
  Fix $v\in\RR$ and let $v_\circ=(v-\mu)/\sigma$. With the fixed selection rate
  $p\in(0,1)$, \eqref{eq:eta-alp} holds, and hence $\tau_\alpha$ in
  Proposition~\ref{prop:s-gaussian} equals $\Phi^{-1}(p)$. Substituting $\tau_\alpha
  = \Phi^{-1}(p)$ into \eqref{eq:gs-cdf} yields
  \[
    \PP(V\leq v\mid D_\alpha=1)
    =
    p^{-1}
    \left[
      \Phi(v_\circ)-\Phi_2\big(v_\circ,-\Phi^{-1}(p);q(\alpha)\big)
    \right],
    \qquad \alpha >0,
    \quad  q(\alpha) \in (-1,1).
  \]
  Differentiation of the bivariate normal density gives
  \(
    \partial\phi_2(x,y;\varrho)/\partial\varrho
    =
    \partial^2\phi_2(x,y;\varrho)/\partial x\,\partial y,
  \)
  and integrating over $(-\infty,x]\times(-\infty,y]$ yields
  \begin{equation}\label{eq:nd}
    \frac{\partial}{\partial\varrho}\Phi_2(x,y;\varrho)
    =
    \phi_2(x,y;\varrho),
    \qquad x,y\in\RR,\quad \varrho\in(-1,1).
  \end{equation}
  Then, the chain rule, \eqref{eq:q-der}, and \eqref{eq:nd} together imply that, for
  $\alpha>0$,
  \begin{align*}
    \frac{\diff}{\diff\alpha}\PP(V\leq v\mid D_\alpha=1)
    &=
    -p^{-1} \phi_2\big(v_\circ,-\Phi^{-1}(p) ; q(\alpha) \big)
    q'(\alpha)
    \\
    &=
    \frac{\alpha\sigma\sigeps^2(1-\rho^2)}{p s(\alpha)^3}\,
    \phi_2\big(v_\circ,-\Phi^{-1}(p) ; q(\alpha)\big) 
    >0.
  \end{align*}
  The strict inequality follows from $p>0$, $\alpha>0$, $|\rho|<1$, and the
  positivity of the bivariate normal density at every finite pair of arguments
  when $q(\alpha)\in(-1,1)$. This proves the first inequality in \eqref{eq:fp}.
  Similarly, substituting $\tau_\alpha = \Phi^{-1}(p)$ into \eqref{eq:gs-mean} gives
  $m(\alpha)=\mu+\sigma q(\alpha)\lambda\big(\Phi^{-1}(p)\big)$. Hence
  \[
    m'(\alpha)
    = 
    \sigma \lambda\big(\Phi^{-1}(p)\big) q'(\alpha)
    =
    - s(\alpha)^{-3}\alpha\sigma^2\sigeps^2(1-\rho^2) \lambda\big(\Phi^{-1}(p)\big)
    <0,
  \]
  which proves the second inequality in \eqref{eq:fp}.
  Finally, substituting the limit in \eqref{eq:q-lim} into \eqref{eq:gs-mean} with
  $\tau_\alpha=\Phi^{-1}(p)$ proves \eqref{eq:fp-lim}.
\end{proof}

Corollary~\ref{cor:s-g-fr} shows that, as the noise scale increases, the cdf
of the true value of the selected units strictly rises at every point, and the
selected mean strictly falls. In particular, the first inequality in \eqref{eq:fp} is
the strict Gaussian counterpart to the distribution-free stochastic-dominance
ordering in Proposition~\ref{prop:fosd}.

Although the selected mean is strictly decreasing in the noise scale, its right
derivative $m'(0+)$ is zero. To see this, note that \eqref{eq:q-der} gives
$q'(0) = 0$ and $q''(0) = -\sigma_\eps^2(1-\rho^2)/\sigma^2$, so
\[
  q(\alpha) = 1 - \frac{\alpha^2\sigma_\eps^2(1-\rho^2)}{2\sigma^2} + o(\alpha^2),
  \qquad \alpha\downarrow 0.
\]
Substituting this into \eqref{eq:gs-mean} with $\tau_\alpha=\Phi^{-1}(p)$ yields
\[
  m(\alpha)
  =
  \mu+\sigma\lambda\big(\Phi^{-1}(p)\big)
  -\frac{\alpha^2\sigeps^2(1-\rho^2)}{2\sigma}
    \lambda\big(\Phi^{-1}(p)\big)
  +o(\alpha^2),
  \qquad \alpha\downarrow0.
\]
Thus the leading decline is quadratic in the noise scale. By \eqref{eq:ld},
$\lambda$ is strictly decreasing, while $\Phi^{-1}$ is strictly increasing.
Hence $\lambda(\Phi^{-1}(p))$ is strictly decreasing in $p$.
Holding the distribution of $(V,\eps)$ fixed, more selective rules
(smaller $p$) therefore have a larger quadratic loss in the selected mean
for the same small increase in noise from $\alpha=0$.

\subsection{Selection at a fixed threshold}\label{sec:gft}

Now, we consider the fixed-threshold setting under Assumption~\ref{ass:normal}.
Under joint normality of $(V,\eps)$, the result in Proposition~\ref{prop:ftv} can be
strengthened. In particular, the threshold-adjusted aggregate value $W(\alpha)$ is
\emph{strictly} decreasing for every $\rho\in(-1,1)$. To see this, fix
$0\leq\alpha_1<\alpha_2$ and define the realized loss
\[
  L_\alpha
  :=|V-\eta|\,1\{D_\alpha\neq D_0\},
  \qquad D_0=1\{V>\eta\}.
\]
For each realization with $V>\eta$, an error can occur only when
$\eps<0$; once increasing $\alpha$ causes rejection, the unit remains
rejected at every larger noise scale. Similarly, when $V<\eta$,
an error can occur only when $\eps>0$; once selected, the unit remains
selected at every larger noise scale. If $V=\eta$, the loss is zero.
Thus $L_{\alpha_2}\geq L_{\alpha_1}$ for every realization.

Now consider the event
\[
  E:=
  \left\{
    V>\eta,\ \eps<0,\quad
    \alpha_1<\frac{V-\eta}{-\eps}<\alpha_2
  \right\}.
\]
Under Assumption~\ref{ass:normal}, $\PP(E)>0$ because the nondegenerate
joint Gaussian density is positive everywhere. On $E$, $V+\alpha_1\eps>\eta $ and
$V+\alpha_2\eps<\eta$. Consequently, $D_{\alpha_1}=D_0=1$ and $D_{\alpha_2}=0$,
giving $L_{\alpha_1}=0$ and $L_{\alpha_2}=V-\eta>0$. Using
$W(0)-W(\alpha)=\E(L_\alpha)$, we obtain
\[
  W(\alpha_1)-W(\alpha_2)
  =\E(L_{\alpha_2}-L_{\alpha_1})
  \geq\E\big[(V-\eta)1_E\big]
  >0.
\]
Hence $W(\alpha)$ is strictly decreasing on $[0,\infty)$.

In the rest of this section, we focus on the analysis of how the selected mean 
$m(\alpha)$ responds to increasing noise scale $\alpha>0$.
In the current setting, $\eta(\alpha)=\eta$, so $\tau_\alpha$ defined in
Proposition~\ref{prop:s-gaussian} becomes 
\begin{equation}\label{eq:tau-ft}
  \tau_\alpha = \frac{\mu-\eta}{s(\alpha)}, \qquad \alpha \geq 0.
\end{equation}
In particular, $\tau_0=(\mu-\eta)/\sigma$. It follows from \eqref{eq:gs-mean},
\eqref{eq:gs-mean-zero}, and $q(0)=1$ that
\begin{equation}\label{eq:evd}
  m(\alpha)=\mu+\sigma q(\alpha)\lambda(\tau_\alpha),
  \qquad \alpha \geq 0.
\end{equation}

\subsubsection{Derivative condition for the selected mean}

We first establish the properties of the function $\psi$ used in
Proposition~\ref{prop:gcs} in the following lemma, and then prove
Proposition~\ref{prop:gcs}.

\begin{lemma}\label{lem:s-psi}
  The function $\psi(z)=z[z+\lambda(z)]$ is strictly increasing from $\RR$
  onto $(-1,\infty)$.
\end{lemma}

\begin{proof}
  Let $Z$ be a standard normal random variable and, for each $z\in\RR$, let $X_z$
  have the conditional distribution of $Z+z$ given $Z>-z$. By \eqref{eq:tn}, 
  $\E X_z=z+\lambda(z)$ and
  \begin{equation}\label{eq:xpsi}
    \E X_z^2 
    = \Var(X_z) + (\E X_z)^2 
    = 1-\lambda(z)[z+\lambda(z)] + [z+\lambda(z)]^2
    = 1+\psi(z).
  \end{equation}
  For each fixed $z\in\RR$, the density of $X_z$ is $f(x,z):=\phi(x-z)/\Phi(z)$ for
  $x\in(0,\infty)$, so \eqref{eq:xpsi} gives
  \begin{equation}\label{eq:xpsi2}
    \psi(z) = \int_0^\infty x^2 f(x,z) \diff x - 1.
  \end{equation}
  Using $\phi'(z) = -z\phi(z)$, we obtain
  \begin{equation}\label{eq:f'}
    \partial_z f(x,z)
    = 
    \frac{(x-z)\phi(x-z)}{\Phi(z)} - \frac{\phi(x-z)}{\Phi(z)} \lambda(z)
    =
    [x-z-\lambda(z)]f(x,z)
    =
    (x-\E X_z)f(x,z).
  \end{equation}
  On every bounded interval of $z$, there is a constant $C>0$ such that
  \[
    x^2|\partial_z f(x,z)| \leq C(1+x^3)\mathrm{e}^{-x^2/4}, 
    \qquad x>0.
  \]
  The right-hand side is integrable, so dominated convergence justifies the
  interchangeability of differentiation and integration. Then, taking derivative on
  both sides of \eqref{eq:xpsi2} with respect to $z$ and using \eqref{eq:f'} yields
  \[
    \psi'(z)
    =\int_0^\infty x^2 \partial_z f(x,z)\,\diff x
    =\int_0^\infty x^2(x-\E X_z)f(x,z)\,\diff x
    =\E[X_z^2(X_z-\E X_z)]
    =\Cov(X_z^2,X_z).
  \]
  Recall that if $X$ is a random variable and $X_\circ$ is an independent copy of
  $X$, then $\Cov(f(X),g(X)) = \E[(f(X)-f(X_\circ))(g(X)-g(X_\circ)]/2$, for any
  functions $f$ and $g$, provided the relevant expectations are finite.
  Consequently, for an independent copy $X_{z,\circ}$ of $X_z$,
  \[
    \Cov(X_z^2,X_z)
    =
    \frac12\E\big[(X_z^2-X_{z,\circ}^2)(X_z-X_{z,\circ})\big]
    =
    \frac12\E\big[(X_z+X_{z,\circ})(X_z-X_{z,\circ})^2\big]>0,
  \]
  where the inequality holds because $X_z$ is strictly positive and nondegenerate.
  This shows that $\psi'(z)>0$ for all $z\in\RR$, and hence $\psi$ is strictly
  increasing on $\RR$.

  To deduce the range of $\psi$, first note that $\psi(z) = z^2+z\lambda(z)>z^2$ for
  $z>0$, and hence $\psi(z)\to\infty$ as $z\to\infty$. For the behavior of $\psi(z)$
  when $z<0$, note that \eqref{eq:xpsi2} gives
  \(
    1+\psi(z) = \int_0^\infty x^2\phi(x-z)\diff x /\Phi(z).
  \)
  Since $\Phi(z) = 1-\Phi(-z) = \int_{-z}^\infty\phi(y)\diff y = \int_0^\infty
  \phi(x-z)\diff x$, it follows that 
  \[
    1 + \psi(z) 
    =
    \frac{\int_0^\infty x^2\phi(x-z)\diff x}
    {\int_0^\infty\phi(x-z)\diff x} 
    =
    \frac
    {\int_0^\infty x^2 \mathrm{e}^{zx-x^2/2}\diff x}
    {\int_0^\infty \mathrm{e}^{zx-x^2/2}\diff x},
    \qquad z\in\RR.
  \]
  For $z<0$, applying the change of variables $u=-zx$ gives
  \[
    z^2[1 + \psi(z)]
    =
    \frac{\int_0^\infty u^2 \mathrm e^{-u-u^2/(2z^2)}\,\diff u}
         {\int_0^\infty \mathrm e^{-u-u^2/(2z^2)}\,\diff u} 
    \to 2,
    \qquad \text{as }z\to-\infty,
  \]
  Hence $\psi(z)\to-1$ as $z\to-\infty$. Together with continuity and
  strict monotonicity, these limits establish that the range of $\psi$ is
  $(-1,\infty)$.
\end{proof}

\begin{proof}[Proof of Proposition~\ref{prop:gcs}]
  Differentiating \eqref{eq:evd} with respect to $\alpha$ gives 
  \begin{equation}\label{eq:dE}
    m'(\alpha)
    =
    \sigma\left[
      q'(\alpha) \lambda(\tau_\alpha) + q(\alpha) \lambda'(\tau_\alpha)
      \diff\tau_\alpha/\diff\alpha
    \right].
  \end{equation}
  where $\tau_\alpha = (\mu-\eta)/s(\alpha)$ was defined in \eqref{eq:tau-ft}.
  Using \eqref{eq:ds}, we obtain
  \begin{equation} \label{eq:dtau}
    \frac{\diff\tau_\alpha}{\diff\alpha}
    =
    -\frac{\tau_\alpha s'(\alpha)}{s(\alpha)}
    =
    -\frac{\tau_\alpha \sigma \sigeps (r_\alpha+\rho)}{s(\alpha)^2}.
  \end{equation}
  Substituting \eqref{eq:q-der}, \eqref{eq:ld}, \eqref{eq:dtau} and the expression of
  $q(\alpha)$ in \eqref{eq:param2} into \eqref{eq:dE}, and using 
  $\psi(z):=z[z+\lambda(z)]$ yields
  \begin{equation} \label{eq:sgd}
    m'(\alpha)
    =
    s(\alpha)^{-3}\sigma^3\sigeps\lambda(\tau_\alpha)
    [C_\alpha\psi(\tau_\alpha)-B_\alpha],
  \end{equation}
  where
  \begin{equation}\label{eq:bc}
    B_\alpha=r_\alpha(1-\rho^2),\qquad 
    C_\alpha=(1+\rho r_\alpha)(r_\alpha+\rho).
  \end{equation}
  Since $s(\alpha)^{-3}\sigma^3\sigeps\lambda(\tau_\alpha)>0$, \eqref{eq:sgd} implies
  that $m'(\alpha)$ has the sign of $C_\alpha\psi(\tau_\alpha)-B_\alpha$.
  Given \eqref{eq:tau-ft}, this proves the sign identity \eqref{eq:gms}.

  Finally, note that for $\alpha>0$, $r_\alpha>0$ and
  $(1+\rho r_\alpha)(r_\alpha + \rho)>0$ if $\rho>-\min\{r_\alpha,
  1/r_\alpha\}$. Thus $m'(\alpha)>0$ if and only if
  $\psi(\tau_\alpha)>B_\alpha/C_\alpha$. Applying Lemma~\ref{lem:s-psi}
  and $p(\alpha)=\Phi(\tau_\alpha)$ proves \eqref{eq:grc}.
\end{proof}

\subsubsection{Shape of the selected mean under nonnegative correlation}

Define
\begin{equation}\label{eq:s-cutoff-index}
  K(r,\rho):=\frac{r(1-\rho^2)}{(1+\rho r)(r+\rho)},
  \qquad r\geq0, \quad \rho\in(-1,1),
\end{equation}
whenever the denominator is nonzero. 

\begin{proof}[Proof of Proposition~\ref{prop:shape}]
  For $\alpha>0$ and $\rho\geq0$, $B_\alpha$ and $C_\alpha$ defined in \eqref{eq:bc}
  are both strictly positive, and $K(r_\alpha,\rho) = B_\alpha/ C_\alpha > 0$.
  Consequently, \eqref{eq:sgd} and Lemma~\ref{lem:s-psi} imply
  \begin{equation}\label{eq:sgn}
    \sgn \big(m'(\alpha) \big)
    =
    \sgn \big( \psi(\tau_\alpha) - K(r_\alpha,\rho)\big)
    =
    \sgn \big( \tau_\alpha - \psi^{-1}(K(r_\alpha,\rho))\big),
  \end{equation}
  where
  \begin{equation}\label{eq:tau0}
    \tau_\alpha 
    = \frac{\mu-\eta}{s(\alpha)}
    = \frac{\mu-\eta}{\sigma\sqrt{1+2\rho r_\alpha + r_\alpha^2}}
    = \frac{\tau_0}{\sqrt{1+2\rho r_\alpha + r_\alpha^2}},
    \qquad
    \tau_0 = \frac{\mu-\eta}\sigma.
  \end{equation}
  For $\rho\geq0$ and $r>0$, define
  \begin{equation}\label{eq:dH}
    d_\rho(r):=\sqrt{1+2\rho r+r^2},
    \qquad
    H_\rho(r):=d_\rho(r)\psi^{-1}\!\bigl(K(r,\rho)\bigr).
  \end{equation}
  Since $p(0)=\Phi(\tau_0)$ and $\Phi$ is strictly increasing,
  \eqref{eq:sgn}--\eqref{eq:dH} imply
  \begin{equation}\label{eq:ms}
    \sgn \big(m'(\alpha)\big)
    =
    \sgn \big(\tau_0-H_\rho(r_\alpha) \big)
    =
    \sgn \big(p(0)-\Phi(H_\rho(r_\alpha)) \big),
    \qquad \alpha>0.
  \end{equation}
  Thus the shape of $m(\alpha)$ is determined by where $\Phi(H_\rho(r_\alpha))$
  crosses the fixed noise-free selection rate $p(0)$.

  In the following analysis, we note that $\psi^{-1}(x)>0$ for $x>0$, because
  $\psi(0)=0$ and $\psi$ is strictly increasing as shown by Lemma~\ref{lem:s-psi}.

  We first consider $\rho=0$. In this case, $K(r,0)=1$ and
  $H_0(r)=\sqrt{1+r^2}\,\psi^{-1}(1)$ for $r>0$.
  Since $\psi^{-1}(1)>0$, $H_0$ is strictly increasing, with limits $\psi^{-1}(1)$ as
  $r\downarrow0$ and $\infty$ as $r\to\infty$.
  Consequently, $\Phi(H_0(r_\alpha))$ is strictly increasing in $\alpha$, with 
  \[
    \lim_{\alpha\downarrow0} \Phi(H_0(r_\alpha)) 
    = \Phi(\psi^{-1}(1)) =: p^\dagger\approx 0.7995,
    \qquad
    \lim_{\alpha\to\infty} \Phi(H_0(r_\alpha)) = 1.
  \]
  If $p(0) \leq p^{\dagger}$, then $p(0) < \Phi(H_0(r_\alpha))$ for every $\alpha>0$.
  By \eqref{eq:ms}, $m'(\alpha)<0$ throughout, so $m(\alpha)$ is strictly decreasing.
  If $p(0)>p^\dagger$, there is a unique $\alpha^*>0$ satisfying
  \(
    \Phi(H_0(r_{\alpha^*})) = p(0).
  \)
  Solving this equation gives
  \begin{equation}\label{eq:a*}
    \alpha^* = \frac{\sigma}{\sigma_\eps}
    \sqrt{\left(\frac{\Phi^{-1}(p(0))}{\Phi^{-1}(p^\dagger)}\right)^2-1}.
  \end{equation}
  By \eqref{eq:ms},
  $m'(\alpha) >0$ for $0< \alpha < \alpha^*$,
  $m'(\alpha) =0$ for $\alpha = \alpha^*$, and
  $m'(\alpha) <0$ for $\alpha > \alpha^*$. 
  Thus $m(\alpha)$ is single-peaked, attaining its unique maximum at $\alpha^*$.
  Finally, the first equality in \eqref{eq:sgn} gives
  \(
    \psi(\tau_{\alpha^*}) = K(r_{\alpha^*},0) = 1,
  \)
  so
  \begin{equation}\label{eq:p*}
    p(\alpha^*) = \Phi(\tau_{\alpha^*}) = \Phi(\psi^{-1}(1)) = p^\dagger.
  \end{equation}
  This proves part~(i) of the proposition.

  Next, consider $0<\rho<1$. For $r>0$, $K(r,\rho)>0$ and $\psi^{-1}(K(r,\rho))>0$,
  and hence $H_\rho(r)>0$. Differentiating $d_\rho(r)$ and $K(r,\rho)$ with respect
  to $r$ gives
  \[
    d_\rho'(r) = \frac{r + \rho}{d_\rho(r)} > 0,
    \qquad
    \partial_r K(r,\rho) 
    =
    \frac{\rho(1-\rho^2)(1-r^2)}{(1+\rho r)^2(r+\rho)^2},
    \qquad 
    r>0.
  \]
  For $r\in(0,1]$, $\partial_r K(r,\rho) \geq 0$, and hence
  \[
    H_\rho'(r) 
    =
    d_\rho'(r)\psi^{-1}(K(r,\rho)) 
    +
    d_\rho(r) \frac{\partial_r K(r,\rho)}{\psi'(\psi^{-1}(K(r,\rho)))}
    >0.
  \]
  For $r>1$, $\partial_r K(r,\rho)<0$. In this case, write
  $z=\psi^{-1}(K(r,\rho))>0$. Hence $K(r,\rho) = \psi(z)$ and 
  \[
    \frac{H_\rho'(r)}{H_\rho(r)}
    =
    \frac{d_\rho'(r)}{d_\rho(r)}
    +
    \frac{\partial_r K(r,\rho)}{K(r,\rho)}
    \frac{\psi(z)}{z\psi'(z)}.
  \]
  Let $g(z) := z+\lambda(z)$ so that $\psi(z) = zg(z)$ and $\psi'(z) = g(z) +
  zg'(z)$. By \eqref{eq:mv}, $g'(z)>0$, so
  \[
    \frac{z\psi'(z)}{\psi(z)}=1+\frac{zg'(z)}{g(z)}>1,
    \qquad z>0
  \]
  Therefore, for $r>1$,
  \[
    \frac{H_\rho'(r)}{H_\rho(r)}
    >
    \frac{d_\rho'(r)}{d_\rho(r)}+\frac{\partial_r K(r,\rho)}{K(r,\rho)}
    =
    \frac
    {r(1+\rho r)(r+\rho)^2+\rho(1-r^2)d_\rho^2(r)}
    {r(1+\rho r)(r+\rho)d_\rho^2(r)}.
  \]
  The last expression is strictly positive because its denominator is positive and
  its numerator simplifies to $r^3+\rho(2+\rho^2)r^2+3\rho^2r+\rho>0$. Hence we have
  shown that $H_\rho'(r) > 0$ for $r>0$, that is, $H_\rho$ is strictly increasing for
  every $\rho\in(0,1)$.

  Fix $\rho\in(0,1)$. As $r\downarrow0$, $K(r,\rho)\to0$ and $d_\rho(r)\to1$, so
  \(
    \lim_{r\downarrow 0} H_\rho(r) = 0.
  \)
  As $r\to\infty$, $K(r,\rho)\sim(1-\rho^2)/(\rho r)$ and $d_\rho(r)\sim r$.
  Moreover, since $\psi(z)/z\to\lambda(0)=\sqrt{2/\pi}$ as $z\downarrow0$, we have
  $\psi^{-1}(z)\sim z/\sqrt{2/\pi}$ as $z\downarrow 0$. Consequently,
  \[
    \lim_{r\to\infty} H_\rho(r)
    =
    \lim_{r\to\infty} d_\rho(r)\psi^{-1}(K(r,\rho))
    =
    \frac{1-\rho^2}{\rho\sqrt{2/\pi}} > 0.
  \]
  Consequently, $\Phi(H_\rho(r_\alpha))$ is strictly increasing in $\alpha$, with 
  \[
    \lim_{\alpha\downarrow0} \Phi(H_\rho(r_\alpha)) = \Phi(0) = \frac12,
    \qquad
    \lim_{\alpha\to\infty} \Phi(H_\rho(r_\alpha)) 
    =
    \Phi\left(\frac{1-\rho^2}{\rho\sqrt{2/\pi}}\right) =: \bar p_\rho > \frac12.
  \]
  If $p(0) \leq 1/2$, then $p(0) < \Phi(H_\rho(r_\alpha))$ for $\alpha>0$;
  by \eqref{eq:ms}, $m'(\alpha)<0$ throughout, so $m(\alpha)$ is strictly decreasing.
  If $p(0) \geq \bar p_\rho$, then $p(0) > \Phi(H_\rho(r_\alpha))$ for $\alpha>0$;
  by \eqref{eq:ms}, $m'(\alpha)>0$ throughout, so $m(\alpha)$ is strictly increasing.
  If $1/2<p(0)<\bar p_\rho$, there is a unique $\alpha_\rho^*>0$ satisfying
  $\Phi(H_\rho(r_{\alpha_\rho^*})) = p(0)$. In particular, 
  \begin{equation}\label{eq:a-rho}
    \alpha_\rho^* = \frac{\sigma}{\sigma_\eps} H_\rho^{-1}\left(\Phi^{-1}(p(0))\right).
  \end{equation}
  By \eqref{eq:ms},
  $m'(\alpha) >0$ for $0< \alpha < \alpha^*_\rho$,
  $m'(\alpha) =0$ for $\alpha = \alpha^*_\rho$, and
  $m'(\alpha) <0$ for $\alpha > \alpha^*_\rho$. 
  Thus, in this case, $m(\alpha)$ is single-peaked, attaining its unique maximum at
  $\alpha^*_\rho$. At the unique maximum,
  \[
    p(\alpha^*_\rho)
    =
    (\Phi\circ\psi^{-1})(K(r_{\alpha^*_\rho},\rho))
    \in(1/2,\,p^\dagger).
  \]
  This completes the proof of part~(ii) of the proposition.
\end{proof}

Parts~(i) and~(ii) of Proposition~\ref{prop:shape} may appear puzzling: for fixed
$p(0)\in(1/2,p^\dagger]$, $m(\alpha)$ is strictly decreasing for $\rho=0$ but
single-peaked for sufficiently small $\rho>0$ such that $p^\dagger<\bar p_\rho$.
This difference can be resolved by showing that $\alpha_\rho^*\to0$ as $\rho\downarrow0$.
To see this, for any fixed $\alpha>0$,
\[
  \lim_{\rho\downarrow0} H_\rho(r_\alpha)
  = 
  H_0(r_\alpha)
  =
  \sqrt{1+r_\alpha^2}\,\Phi^{-1}(p^\dagger)
  >
  \Phi^{-1}(p^\dagger)
  \geq
  \Phi^{-1}(p(0)).
\]
By \eqref{eq:a-rho}, $\Phi^{-1}(p(0))=H_\rho(r_{\alpha_\rho^*})$. 
Thus $H_\rho(r_\alpha)>H_\rho(r_{\alpha^*_\rho})$ for all sufficiently small
$\rho>0$. Since $H_\rho$ is strictly increasing, this means $0<\alpha_\rho^*<\alpha$
for all sufficiently small $\rho>0$. As $\alpha>0$ is arbitrary, we must have
$\alpha_\rho^*\to0$ as $\rho\downarrow0$.
Joint continuity of the Gaussian selected-mean formula in
$(\alpha,\rho)$ then gives $m(\alpha_\rho^*)-m(0)\to0$.
Thus the initial increase disappears in the independence limit.

\subsubsection{Selected mean under negative correlation}

Proposition~\ref{prop:gcs} applies to every $\rho\in(-1,1)$, but the global shape
classification in Proposition~\ref{prop:shape} requires $\rho\geq0$. For $-1<\rho<0$,
we show that the selected mean initially rises when $p(0)<1/2$ and give a
separate example with a local minimum followed by a local maximum.

\paragraph{Initial response.}
By continuity, \eqref{eq:sgd} extends to the right derivative at zero. Since
$s(0)=\sigma$, $B_0=0$, and $C_0=\rho$, it gives
\(
  m'(0+)=\sigeps\rho\lambda(\tau_0)\psi(\tau_0),
\)
implying that
\(
  \sgn(m'(0+)) = -\sgn(\psi(\tau_0))
\)
given $\rho<0$.
By Lemma~\ref{lem:s-psi} and $\psi(0)=0$, we have $\sgn(\psi(\tau_0)) =
\sgn(\tau_0)$. Furthermore, $\tau_0=\Phi^{-1}(p(0))$ implies that
$\sgn(\tau_0)= \sgn(p(0)-1/2)$. Consequently,
\[
  \sgn(m'(0+)) = -\sgn(p(0)-1/2).
\]
This means that $m'(0+)>0$ when $p(0)<1/2$ and $m'(0+)<0$ when $p(0)>1/2$.
Continuity of the derivative therefore implies strict increase or decrease,
respectively, for sufficiently small positive noise scales. If $p(0)=1/2$,
then $p(\alpha)=1/2$ throughout and
$m(\alpha)=\mu+\sigma q(\alpha)\lambda(0)$ is strictly decreasing.

The initial rise in the selected mean when $p(0)<1/2$ operates through the selection
rate. Negative correlation implies $s'(0)=\rho\sigeps<0$, so increasing the noise
scale initially reduces the assessment standard deviation. When $p(0)<1/2$,
\eqref{eq:p'} gives $p'(0)<0$. Since $q(0)=1$ and $q'(0)=0$, \eqref{eq:gsm} implies
that the first-order increase in the selected mean comes entirely from this decline
in the selection rate. The value--assessment correlation $q(\alpha)$ nevertheless
strictly decreases with noise, exerting a countervailing effect for every
$\alpha>0$.

\paragraph{Three noise regions.}
The assessment standard deviation reaches its minimum at $r_\alpha=-\rho$,
whereas the value--assessment correlation crosses zero at
$r_\alpha=-1/\rho$. Since $0<-\rho<-1/\rho$, these points define three cases: 
\begin{enumerate}
  \item For $0<r_\alpha<-\rho$, $s'(\alpha)<0$ and $q(\alpha)>0$.
  \item For $-\rho<r_\alpha<-1/\rho$, $s'(\alpha)>0$ and $q(\alpha)>0$.
  \item For $r_\alpha>-1/\rho$, $s'(\alpha)>0$ and $q(\alpha)<0$.
\end{enumerate}
In Case 2, $C_\alpha>0$, so $m'(\alpha)>0$ if and only if condition 
\eqref{eq:grc} holds. If $r_\alpha$ equals $-\rho$ or $-1/\rho$, $C_\alpha=0$ and
$B_\alpha>0$, so \eqref{eq:sgd} implies $m'(\alpha)<0$.

In Cases 1 and 3, $C_\alpha<0$, and \eqref{eq:sgd} implies that
$m'(\alpha)>0$ if and only if $\psi(\tau_\alpha)<K(r_\alpha,\rho)$.
Since $\psi(\RR)=(-1,\infty)$, $m'(\alpha)<0$ if $K(r_\alpha,\rho)\leq-1$.
If $-1<K(r_\alpha,\rho)<0$, then $m'(\alpha)>0$ if and only if
\[
  p(\alpha)<\Phi\!\left(\psi^{-1}\!\bigl(K(r_\alpha,\rho)\bigr)\right).
\]
In Case 1, a falling selection rate raises the selected mean because $q(\alpha)>0$;
in Case 3, a rising selection rate raises the selected mean because $q(\alpha)<0$.
In either case, this effect must outweigh the decline in $q(\alpha)$.

\paragraph{A local minimum followed by a local maximum.}
To show why Proposition~\ref{prop:shape} does not extend to negative correlation,
take $\rho=-0.3$ and $p(0)=0.97$, so $\tau_0=\Phi^{-1}(0.97)\in(1.88,1.89)$.
The derivative has the sign of
\[
  F(r):=(1+\rho r)(r+\rho)
    \psi\!\left(\frac{\tau_0}{\sqrt{1+2\rho r+r^2}}\right)
    -r(1-\rho^2).
\]
Here
\[
  \begin{aligned}
    F(0.3)&=-0.273,\\
    F(1)&=0.49\psi\!\left(\frac{\tau_0}{\sqrt{1.4}}\right)-0.91
      >0.35\tau_0^2-0.91>0,\\
    F(3)&=0.27\psi\!\left(\frac{\tau_0}{\sqrt{8.2}}\right)-2.73<-2.19.
  \end{aligned}
\]
The bounds use $\psi(z)>z^2$ for $z>0$ and $\psi(z)<2$ for $0<z<1$;
the latter follows from $0<\lambda(z)<\lambda(0)=\sqrt{2/\pi}<1$.
Continuity therefore gives at least one local minimum for $r\in(0.3,1)$ and
at least one local maximum for $r\in(1,3)$. Consequently, the shape
classification in Proposition~\ref{prop:shape} requires additional cases under negative
correlation.

\section{Subgroup ATT in treatment selection}\label{sec:s-subgroup-att}

This section formalizes the application in Section~\ref{sec:subgroup-att} of the
main paper by comparing two subgroups with different levels of decision noise.
For a representative member of group $g\in\{1,2\}$, let $V_g=Y_g(1)-Y_g(0)$
be the individual treatment gain and let $\eps_g$ be the baseline assessment error.
To isolate the effect of noise on participant composition, suppose
$(V_1,\eps_1)$ and $(V_2,\eps_2)$ have the same joint distribution, with
$\E|V_g|<\infty$. The groups therefore share the population gain distribution
and baseline value--error dependence.

Decision noise in group $g$ is $\alpha_g\eps_g$, with
$0\leq\alpha_1<\alpha_2$, so group 2 has the larger error scale.
Treatment selection follows
\begin{equation}\label{eq:s-group-rule}
  D_g=1\{V_g+\alpha_g\eps_g>\eta_g\},
  \qquad g\in\{1,2\}.
\end{equation}
In the fixed-rate setting, independent randomized tie-breaking is allowed as
specified in the main paper. Define
\begin{equation}\label{eq:s-group-att}
  p_g:=\PP(D_g=1)>0,
  \qquad
  \mathrm{ATT}_g:=\E(V_g\mid D_g=1),
  \qquad
  A_g:=\E(V_gD_g)=p_g\mathrm{ATT}_g.
\end{equation}
\begin{samepage}
Here $A_g$ is the average outcome gain per group member relative to universal
nonparticipation. Comparisons are per group member and do not depend on group sizes.
\par
\end{samepage}

\subsection{Equal subgroup selection rates}

We first consider the case where the selections are at a fixed rate so that $p_1 =
p_2$, while the thresholds may differ across groups. The following corollary shows
that identical subgroup treatment rates can coexist with different ATTs generated
solely by different noise scales.

\begin{corollary}[Equal subgroup rates]\label{cor:gfr}
  Under the subgroup setup above, let $0\leq \alpha_1 < \alpha_2$, and suppose the
  thresholds and any tie-breaking probabilities satisfy $p_1=p_2=p\in(0,1)$. Then the
  distribution of $V_1$ conditional on $D_1=1$ first-order stochastically dominates
  that of $V_2$ conditional on $D_2=1$, and
  \[
    \mathrm{ATT}_2\leq\mathrm{ATT}_1,
    \qquad A_2\leq A_1.
  \]
  If the common distribution of $(V_1,\eps_1)$ and $(V_2,\eps_2)$ satisfies
  Assumption~\ref{ass:normal}, both inequalities are strict.
\end{corollary}

\begin{proof}[Proof of Corollary~\ref{cor:gfr}]
  Because both groups have the same joint distribution, their selection rules
  can be evaluated on a common generic pair $(V,\eps)$. Proposition~\ref{prop:fosd}
  then gives the distributional ordering and the ATT inequality.
  The aggregate-value inequality follows from $A_g=p\mathrm{ATT}_g$.
  Under joint normality, \eqref{eq:gs-mean} gives
  \(
    \mathrm{ATT}_2-\mathrm{ATT}_1
    =\sigma\lambda\bigl(\Phi^{-1}(p)\bigr)
      [q(\alpha_2)-q(\alpha_1)]<0,
  \)
  since $\lambda>0$ and $q$ is strictly decreasing on $[0,\infty)$.
  This proves the corollary.
\end{proof}

\subsection{A common fixed threshold}

Now let $\eta_1=\eta_2=\eta$ and use the strict threshold rule in
\eqref{eq:s-group-rule}. The selection rates may differ. Define
\begin{equation}\label{eq:s-group-value}
  W_g
  := \E[(V_g-\eta)D_g] 
  =  p_g(\mathrm{ATT}_g-\eta),
  \qquad g\in\{1,2\}.
\end{equation}
If $\eta$ is a participation cost excluded from $V_g$, then $W_g$ is average net
surplus per group member.

\begin{corollary}[A common subgroup threshold]\label{cor:gft}
  Under the subgroup setup above, let $0\leq \alpha_1<\alpha_2$, suppose both groups
  face a common threshold $\eta$, and assume $p_1, p_2>0$. Then, $W_2\leq W_1$, and
  \begin{equation}\label{eq:gac}
    \mathrm{ATT}_2>\mathrm{ATT}_1
    \qquad\text{if and only if}\qquad
    (\mathrm{ATT}_1-\eta)(p_1-p_2)>W_1-W_2.
  \end{equation}
  If the common distribution of $(V_1,\eps_1)$ and $(V_2,\eps_2)$ satisfies
  Assumption~\ref{ass:normal}, then $W_2<W_1$.
\end{corollary}

\begin{proof}[Proof of Corollary~\ref{cor:gft}]
  Evaluating both rules on the common distribution with the threshold fixed,
  Proposition~\ref{prop:ftv} gives $W_2\leq W_1$; the strict inequality under
  joint normality was established in Section~\ref{sec:gft}.
  Equation~\eqref{eq:s-group-value} yields
  \(
    p_2(\mathrm{ATT}_2-\mathrm{ATT}_1) 
    =
    (\mathrm{ATT}_1-\eta)(p_1-p_2)-(W_1-W_2).
  \)
  Since $p_2>0$, this proves \eqref{eq:gac}.
\end{proof}

\paragraph{Gaussian comparisons.}
Under Assumption~\ref{ass:normal} and a common threshold,
$\mathrm{ATT}_g=m(\alpha_g)$, and both groups share the noise-free selection rate
$p(0)=\Phi((\mu-\eta)/\sigma)$. Two cases illustrate how higher ATT can accompany
lower net surplus in the noisier group when the threshold is a participation cost.

Under independent errors with $p(0)>p^\dagger$, Proposition~\ref{prop:shape}(i)
gives a unique maximizing scale $\alpha^*$, defined in \eqref{eq:a*}.
For $0\leq\alpha_1<\alpha_2\leq\alpha^*$,
\[
  \mathrm{ATT}_2>\mathrm{ATT}_1,
  \qquad W_2<W_1.
\]
If both scales lie after the maximum, the ATT ordering reverses. If they straddle
it, their order alone does not determine the ATT ranking; \eqref{eq:gac} gives
the exact comparison.

Positive correlation can make this reversal hold for every pair of noise scales:
if $\rho>0$ and $p(0)\geq\bar p_\rho$, Proposition~\ref{prop:shape}(ii) gives
$\mathrm{ATT}_2>\mathrm{ATT}_1$ and $W_2<W_1$ for every
$0\leq\alpha_1<\alpha_2$. For $\rho<0$, Proposition~\ref{prop:gcs} determines the
local response but does not give a global ATT ordering.

These are differences in true average gains among participants, despite identical
gain distributions in the underlying subgroup populations. An ATT gap alone
identifies neither a difference in population gain distributions nor a difference
in decision noise. The assumption of a common distribution isolates the latter mechanism;
empirical subgroup comparisons may also reflect differences in gains, baseline
errors, or thresholds.

\end{document}